\documentclass[letterpaper,10pt,twocolumn]{IEEEtran}

\usepackage{color}

\usepackage{setspace}
\usepackage{comment}
\usepackage{url}

\usepackage{graphicx}
\usepackage{amssymb}
\usepackage{amsmath,amsfonts,amssymb,euscript,epsfig,psfrag,amsthm,enumerate,float,afterpage,subfigure}

\ifcsname appendices\endcsname

\else
  \usepackage{appendix}
\fi

\newtheorem{thm}{Theorem}

\newtheorem{lem}{Lemma}

\usepackage{cite}

\usepackage{rotate}
\usepackage{verbatim}
\usepackage{multirow}
\usepackage{wrapfig}
\usepackage{setspace}
\usepackage{subfig}
\usepackage{multirow}

\begin{document}

\title{\bf \LARGE Data Center Cost Optimization Via Workload Modulation Under Real-World Electricity Pricing}

\author{Cheng Wang, Bhuvan Urgaonkar, Qian Wang, George Kesidis, and Anand Sivasubramaniam\\
Departments of CSE, EE, and MNE\\
The Pennsylvania State University}

\date{}

\maketitle

\singlespacing

\begin{abstract}
We formulate optimization problems to study how data centers might modulate their power demands for cost-effective operation taking into account three key complex features exhibited by real-world electricity pricing schemes: (i) time-varying prices (e.g., time-of-day pricing, spot pricing, or higher energy prices during "coincident" peaks) and (ii) separate charge for peak power consumption. Our focus is on demand modulation at the granularity of an entire data center or a large part of it. For computational tractability reasons, we work with a fluid model for power demands which we imagine can be modulated using two abstract knobs of demand dropping and demand delaying (each with its associated penalties or costs). Given many data center workloads and electric prices can be effectively predicted using statistical modeling techniques, we devise a stochastic dynamic program (SDP) that can leverage such predictive models. Since the SDP can be computationally infeasible in many real platforms, we devise approximations for it. We also devise fully online algorithms that might be useful for scenarios with poor power demand or utility price predictability. For one of our online algorithms, we prove a competitive ratio of $2-\frac{1}{n}$. Finally, using empirical evaluation with both real-world and synthetic power demands and real-world prices, we demonstrate the efficacy of our techniques. As two salient empirically-gained insights: (i) demand delaying is more effective than demand dropping regarding to peak shaving (e.g., 10.74\% cost saving with only delaying vs. 1.45\% with only dropping for Google workload) and (ii) workloads tend to have different cost saving potential under various electricity tariffs (e.g., 16.97\% cost saving under peak-based tariff vs. 1.55\% under time-varying pricing tariff for Facebook workload).
\end{abstract}

\vspace{-0.1in}
\section{Introduction}

It is well-known that data centers incur significant costs towards powering and cooling their computing, networking, and storage (IT) equipment, making the optimization of these costs an important research area~\cite{FanWB07,RenHX12}. A particularly significant contributor to the difficulty of such decision-making is the {\em complexity of real-world electric utility pricing schemes}. Electric utilities employ a variety of pricing mechanisms to disincentivize consumers from posing high or unpredictable power demands (especially simultaneously) or to incentivize consumers to pose predictable demands. Most of these mechanisms employ one (or a combination)  of the following:  (i) {\em peak-based pricing}, wherein a component of the electricity bill is dependent on the peak power drawn over the billing cycle (typically a month)~\cite{duke}, (ii) {\em time-varying pricing}, wherein the per unit energy price fluctuates over time (there are many examples of this operating at different time granularity such as very fine time-scale spot prices~\cite{tou,sceg} or higher prices during periods of ``coincident peaks'' experienced by the utility~\cite{coincidentPeak}), and (iii) {\em load-dependent or tiered pricing}, wherein higher energy prices are applied for higher power demands~\cite{sdgeTier,oebTier,idahoTier}.

Whereas recent research has begun looking at data center cost optimization bearing some of these pricing complexities in mind, this is still a nascent field in many ways. In particular, most instances of such work focus exclusively on a particular  pricing feature (e.g., peak-based~\cite{govindan-asplos2012}, coincident peak pricing~\cite{liu2013}, or real-time prices~\cite{urgaonkar-sigmetrics2011}), whereas real-world tariffs often combine these features (e.g., ~\cite{duke} combines time-of-day with peak pricing and ~\cite{coincidentPeak} considers tiered pricing together with coincident peak), making cost-effective data center operation even more complicated. This is the context of our paper: {\em how should data centers optimize their costs given the various features of real-world electricity pricing?} An important motivation for our study of data center operation under realistic descriptions/models of pricing schemes arises from its potential in informing how certain (particularly big) data centers might  negotiate suitable pricing structures with their electric utility, a question that has begun to receive attention recently~\cite{mama13}.

Broadly speaking, the vast literature on data center power cost optimization may be understood as employing one or both of (i) demand-side (based on modulating the data center's power demand from within) and/or (ii) supply-side (based on employing additional sources of energy supply the data center's existing utility provider and backup sources) techniques.  
Our focus is on techniques based on ``IT knobs,'' a subset of (i) that relies upon software/hardware mechanisms to change the power consumption  of the IT machinery (servers, storage, networking) within the data center (and perhaps indirectly its cooling-related power consumption). Other demand-side modulation based on local generation capabilities or energy storage (including work by some of the co-authors)~\cite{govindan-isca2011} and a large body of work that has emerged on supply-side (also including work by the co-authors)~\cite{RenWUS12} are complementary to our work, and  cost-effective operation using these techniques is interesting future work.


We make the following research contributions:
\begin{itemize}
\item {\em Problem Formulation}: We formulate an optimization problem for a data center that wishes to employ IT knobs for power demand modulation for cost-effective operation. Our key novelty over related work in this area is our incorporation of various features of real-world electricity pricing schemes into a single unified formulation. A second important feature of our formulation is our general representation of the myriad power vs. performance/revenue trade-offs for diverse data center workloads via two power demand modulation knobs, namely demand dropping and demand delaying (with associated penalties or costs).
\item {\em Algorithm Design}: Given that power demands and electricity prices can exhibit uncertainty, besides considering offline algorithms, we devise a stochastic dynamic program (SDP) that leverages predictive workload and price models. We also devise approximations for our SDP that might be useful for scenarios where the exact SDP proves computationally intractable. Finally, for scenarios with poor input predictability, we also explore fully online algorithms, and prove a competitive ratio of $2-\frac{1}{n}$ for one of them.
\item {\em Empirical Analysis}: We evaluate the efficacy of our algorithms in offering cost-effective operation for a variety of workloads derived from real-world traces. Our results help us understand when and why demand modulation is (or is not) useful for cost-efficacy. We find that our stochastic control techniques offer close-to-optimal cost savings when the workloads exhibits strong time-of-day patterns; our online algorithms performs well even when there is unpredictable flash crowd in the workload.
\end{itemize}

The rest of this paper is organized as follows. 
In Section~\ref{sec:back}, we discuss the scope and context of our work, key assumptions,  and its relation with prior work. In Section~\ref{sec:problem}, we present our offline problem formulation. In Section~\ref{sec:stochastic}, we present our stochastic control algorithms. In Section~\ref{sec:online}, we present online algorithms that might be suitable for scenarios with poor predictability. In Section~\ref{sec:eval}, we present our empirical evaluation. Finally, we present concluding remarks in Section~\ref{sec:conclus}.

\vspace{-0.1in}
\section{Background and Related Work}
\label{sec:back}

\subsection{Context and Key Assumptions}
\label{sec:assump}
We consider our problem at the granularity of an entire data center or a ``virtual'' data center (i.e., a subset of the data center whose IT resources are dedicated to a ``tenant'' application that is allowed to carry out its own power cost optimization). 
For reasons of computational feasibility, we make two key assumptions/simplifications. First, we imagine power demands as being ``fluid,'' whereas, in practice, power cost optimization must deal with discrete resource allocations and software job characteristics. 
Second, we choose to capture the IT power/performance/cost trade-offs via two abstract knobs that we consider applying directly to the power demand input: demand dropping and demand delaying. We assume that we can leverage existing or future work on translation mechanisms between our fluid power demand and its modulation via delaying/dropping (on the one hand), and actual IT resources, their control knobs, and software characteristics and performance (on the other hand). 

We employ convex non-decreasing functions to model the loss due to demand delaying ($l_{\sf delay}(demand,delay)$) and dropping ($l_{\sf drop}(demand)$) based on evidence in existing work~\cite{van1995dynamic, freris2010resource}. We denote by $\tau$ the delay tolerance for a unit of delayed power demand. 
 We offer three workload scenarios (whose analogues we evaluate in Section~\ref{sec:eval}) as concrete examples of the validity/suitability of our abstractions of demand dropping and delaying:
\begin{enumerate}
\item {\bf Example 1:} A Web search application with partial execution for meeting response time targets~\cite{XuL2013} is an example of a workload that is delay-sensitive but can shed some of its power needs (relative to that corresponding to the best quality results) to meet delay targets. Another example is a video server that exploits multiple fidelity videos that MPEG allows to guarantee fluent streaming~\cite{freris2010resource}.
\item {\bf Example 2:} For several batch workloads (e.g., a MapReduce-based set of tasks~\cite{yao2012data}), dropping is intolerable but delaying demand as long as it finishes before a deadline is acceptable. 
\item {\bf Example 3:} Some applications can have combinations of the above two. E.g., a search engine typically has a user-facing front-end that is delay-sensitive as well as a backend crawler that is delay-tolerant~\cite{teregowda2010citeseerx}.
\end{enumerate}

In fact, besides the above scenarios, many other systems can also employ various IT knobs (e.g., redirecting requests to other data centers) to save costs, and these can be effectively translated into corresponding delaying and dropping related parameters. The abstractions of demand delaying and dropping help us keep our formulations computationally tractable. 

\vspace{-0.1in}
\subsection{Related Work}
\label{sec:related}
Related work in this area is vast and we discuss representative efforts in the most close research topics.
A large body of work exploits the use of demand-side IT knobs (e.g., Dynamic Voltage-Frequency Scaling (DVFS), admission control, scheduling, migration, etc.). Wang et al.  study a cluster-level power control problem via DVFS based on feedback control~\cite{WangC08}. Lu et al.  reduce idle power costs by dynamically turning off unnecessary servers~\cite{LuCA2013}.  Zhou et al. develop online techniques to make decisions on request admission control, routing and maximize operating profit as well as minimize power consumption~\cite{zhou2013arbitrating}. Closely related to our work, Liu et al. minimize operational costs (together with delay costs) by redirecting requests among geographically located data centers. In face of the load-dependent/location-dependent energy pricing tariff~\cite{LiuLWLA11}. Xu et al. consider optimizing the electricity costs under peak-based pricing by exploiting the ability to partially execution service requests~\cite{XuL2013}. This is very closely related to our problem - partial execution may be one way for a data center to {\em drop} part of its power demand. However, they assume perfect knowledge of demand within a planning period of only one day (billing cycles for peak-based tariff are typically a month). Zhao et al.  also solve the peak-minimizing problem but in the context of EV charging via rescheduling of charging jobs~\cite{zhaopeak}. Finally, Liu et al.  develop online algorithms to avoid drawing power during coincident peaks, when the utility imposes a steep energy cost~\cite{liu2013}. However, coincident peaks occur due to high demands imposed on the utility by a {\em collection} of its customers, whereas the data center's internal peak power consumption - which we are interested in under peak-based tariff - may not occur at the same time as the coincident peaks. All these prior research works either focus on a particular pricing tariff, or one specific control knob, while in this paper we seek the opportunity to optimize costs via various workload modulation knobs under different real-world electricity pricing tariffs. 

A second line of work studies cost optimization using non-IT knobs for demand-side modulation. As one example, recent works have  demonstrated the efficacy of energy storage devices such as uninterruptible supply unit (UPS) batteries in reducing the electricity costs of data centers under peak-based pricing~\cite{govindan-isca2011,govindan-asplos2012,wang-sigmetrics2012}. Similar results have also been shown for home consumers~\cite{barker2012smartcap}. However, while these works are based on offline formulations, our interests are in decision-making in the face of uncertainty.  Bar-Noy et al. develop online algorithms with known competitive ratios for peak minimization using batteries~\cite{barnoy2008,barnoy2008a}. We find this highly complementary to the online algorithms we develop in Section~\ref{sec:online} with the key difference being that their control knobs are based on using the battery while ours are based on demand modulation. Ven et al. consider a stochastic version of the problem of cost minimization using batteries and develop a Markov Decision Process (MDP) as their approach~\cite{hegde11}. This is similar in approach to our stochastic control in Section~\ref{sec:sdp}. However, there are two important differences. First, we handle peak-based and time/load-dependent prices in a unified SDP, while their work focusses only on time-varying prices (for which the DP is easier to cast). Secondly, we also address the scalability problems that our SDP might pose and devise approximation algorithms to overcome them. Urgaonkar et al. develop an  online Lyapunov optimization based technique to minimize operational cost under time-varying electricity tariff~\cite{urgaonkar-sigmetrics2011}. This approach is very appealing, especially since it does not require any prior knowledge of future information. Unfortunately, to the best of our knowledge, it is only useful to derive algorithms that asymptotically optimize objectives based on averages. The applicability of this approach to the peak-based elements within real-world pricing schemes is unclear.

Finally, growing interest has appeared in power management of smart grid and data centers, employing supply-side knobs, such as renewable energy, diesel generator, etc. Researchers have been looking at how to reduce carbon emission and electricity bill of data centers at the same time. ~\cite{RenWUS12} develops offline algorithm to evaluate the cost saving potential of introducing on-site/off-site renewable energy for data centers. ~\cite{deng2013multigreen} and ~\cite{LuTCCL2013} consider minimizing the long-term supply cost under time-varying pricing tariff and develop algorithms to schedule the online generation of renewable energy. Such supply-side techniques are complementary to our research.

\section{Problem Formulation}
\label{sec:problem}


\vspace{-0.1in}
\subsection{Notation}

%
%

\begin{wrapfigure}[11]{r}{0.25\textwidth}
  \centering
  \includegraphics[width=0.23\textwidth]{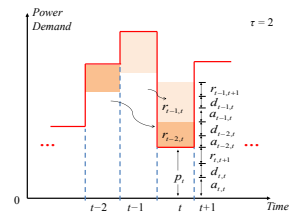}
  \caption{\small An illustration of decision variables in our formulation.}
\label{fig:decisionVar}
\end{wrapfigure}

\noindent{\bf Input Parameters:} We consider a discrete-time model wherein control decisions are made at the beginning of {\em control windows} of equal duration $\delta$. We denote by $T$ the number of such control windows within a single billing cycle, which constitutes our {\em optimization horizon}. 
Let the time-series $\{p_t: 1 \le t \le T\}$ denote the power demand of the (possibly virtual, as described in Section~\ref{sec:back}) data center over the optimization window of interest, with $0 \le p_t \le p_{\max}$ denoting its power demand during the $t^{\rm th}$ control window. 
Here $p_{\max}$ denotes the maximum power demand the data center may pose during a single control window and would depend on the data center's workload, its IT/cooling infrastructure, and the power delivery infrastructure provisioned by the data center. To represent the time-varying pricing employed in many electric tariffs (e.g.,~\cite{tou,sceg}), we define 
a time-varying energy price of $\alpha_t$ dollars per unit energy. To represent peak-based charging (e.g.,~\cite{duke,govindan-isca2011,XuL2013}), we define a peak price of $\beta$ dollars per unit peak power per billing cycle. 
Finally, our notation and assumptions for demand modulation costs are as introduced in Section~\ref{sec:back}. Table~\ref{tbl:params} in Section~\ref{sec:eval} summarizes the values we use for these parameters in our evaluation. 


\noindent{\bf Decision Variables:} Given the possibility of delaying portions of its demand, at the beginning of control window $t$ (or simply $t$ henceforth), the data center can have ``residual'' demands from the past (in addition to the newly incoming demand $p_t$). Given that demands may not be delayed by more than $\tau$ windows, these residual demands may have originated in (i.e., first postponed during) the set of windows $h(t) = \{\max\{1,t-\tau\},..., t-1\}$. We denote these residual demands as the set $\mathcal{R}_t=\{r_{i,t}\}_{i \in h(t)}$. We denote the peak demand admitted in any window during $[1, ..., T]$ as $y_{\rm max}$. The control actions to be taken by the data center during $t$ involve admitting, postponing, and dropping portions of $\mathcal{R}_t$ and $p_t$. We denote as $a_{i,t}$ ($i \in h(t)$) the portion of residual demand $r_{i,t}$ that is admitted during $t$; let $a_{t,t}$ denote the portion of $p_t$ that is admitted during $t$ (i.e., without incurring any delay). We denote as $h^+(t)$ the set $h(t)\cup \{t\}$. Let $a_t^+=\sum_{i \in h^+(t)} a_{i,t}$ denote the total demand admitted during $t$. Similarly, we denote as $d_{i,t}$ ($i \in h(t)$) the portion of residual demand $r_{i,t}$ that is dropped during $t$; let $d_{t,t}$ denote the portion of $p_t$ that is dropped during $t$ (i.e., without incurring any delay). Finally, we denote as $\mathcal{A}_t$ the set $\{a_{i,t}\}_{i \in h^+(t)}$, and as $\mathcal{D}_t$ the set $\{d_{i,t}\}_{i \in h^+(t)}$ of decision variables. Figure~\ref{fig:decisionVar} helps explain our notation using $\tau=2$.


\vspace{-0.1in}
\subsection{Offline Decision Making}
\label{sec:offline}


\noindent{\bf Objective:} We choose as our data center's objective the minimization of the sum of its utility bill and any revenue loss resulting from demand modulation:

\noindent $\mathcal{O}(\{\mathcal{A}_t\},\{\mathcal{D}_t\}) \stackrel{\text{def}}{=} \beta y_{\rm max} + \sum_t \Big(\alpha_t a_t^+  + \sum_{i \in h^+(t)} l_{\sf drop}(d_{i,t}) +  \sum_{i \in h(t)} l_{\sf delay}(a_{i,t},t-i)\Big)$.   



Notice that our objective does not have an explicit revenue term which might give the impression that it does not capture the data center's incentive for admitting demand. It is important to observe that the incentive for admitting demand comes from the costs associated with dropping or delaying demand. Generally, one would have $l_{\sf drop}(a) > \alpha_t(a)\cdot a$ to disallow scenarios where the data center prefers dropping all demand to admitting it. However, there can be situations (e.g., extremely high energy prices during coincident peaks~\cite{coincidentPeak,liu2013}) when this is not so. Finally, we could also include a revenue model of accepted demand, simply resulting in $\alpha_t(.)<0$ in our problem formulation.

\noindent{\bf Constraints: } The aggregate demand in the data center at the beginning of $t$ is given by the new demand $p_t$, and any demand unmet so far (deferred from previous time slots $h(t)$) $\mathcal{R}_t$. Since this demand must be treated via a combination of the following three: (i) serve demand $a_{t,t}$, (ii) drop demand $d_{t,t}$, and (iii) postpone/delay demand ($r_{t,t+1}$) (to be served during $[t+1, ..., T]$), we have:
\begin{equation}
\label{eq:eqone}
p_t - a_{t,t} - d_{t,t} = r_{t,t+1}, ~\forall t.
\end{equation}
The residual demand from $i \in h(t)$ that is not admitted during $t$ is either dropped during $t$ or postponed to the next time slot $t+1$:
\begin{equation}
r_{i,t} - a_{i,t} - d_{i,t} = r_{i,t+1}, i \in h(t), ~\forall t.
\end{equation}
Any demand that has been postponed for $\tau$ time slots may not be postponed any further:
\begin{equation}
r_{t-\tau,t} - a_{t-\tau,t} - d_{t-\tau,t} = 0, ~\forall t.
\end{equation}
To keep our problem restricted to one billing cycle, we add an additional constraint that any delayed demand (even if it has been postponed for less than $\tau$ time slots) must be admitted by the end of our optimization horizon:
\begin{equation}
r_{i,T+1} = 0, i \in h(t).
\end{equation}
Alternate formulations that minimize costs over multiple billing cycles may relax the constraint above. The peak demand admitted $y_{\rm max}$ must satisfy the following:
\begin{equation}
y_{\rm max} \ge a_t^+,  ~\forall t.
\end{equation}
Finally, we have:
\begin{equation}
\label{eq:eqlast}
 y_{\rm max}, a_{i,t}, d_{i,t}, r_{i,t} \geq 0, i \in h^+(t), ~\forall t.
\end{equation}

\noindent{\bf Offline Problem:} Based on the above, we formulate the following offline problem (called OFF):
\begin{align*}
\textrm{Minimize}    ~~~~~~~~~~& \mathcal{O}(\{\mathcal{A}_t\},\{\mathcal{D}_t\}), \\
\textrm{subject to}  ~~~~~~~~~~& (\ref{eq:eqone})-(\ref{eq:eqlast}).
\end{align*}

\vspace{-0.1in}
\subsection{Discussion}

\begin{enumerate}
\item Owing to the following lemma, we can simplify our problem formulation somewhat by setting $d_{i,t}=0, i \in h(t), \forall t$):
\begin{lem}
\label{lem:DelayNotDrop}
There always exists some optimal solution of OFF that never postpones some demand only to drop it in the future if $l_{\sf delay}(,) > 0$.
\end{lem}
Our proof is based on a relatively straightforward contradiction-based argument. We omit it here for space and present it in the Appendix. When devising our stochastic optimization formulation in Section~\ref{sec:stochastic}, we exploit this lemma allowing us to express our program with fewer variables and constraints.
\item  Under convexity assumptions for the functions $l_{\sf drop}(.)$ and $l_{\sf delay}(.)$, which we justified in Section~\ref{sec:back}, OFF is a computationally tractable convex program. To capture tiered pricing, our formulation can be easily extended by defining energy price as a function $\alpha_t(a)$  of the admitted demand $a$ during a billing cycle. Since tiered prices are typically non-continuous functions of admitted demand (e.g., $0.078\$/kWh$ for demand up to $750kW$ and then a jump to $0.091\$/kWh$ for demands exceeding $750kWh$~\cite{oebTier}), this aspect of our formulation can render it a non-convex and computationally difficult program. However, a dynamic programming approach (similar to the SDP in Section~\ref{sec:sdp}) can still be applied to evaluate tiered pricing by defining states appropriately.
\item Our problem formulation is general enough to incorporate the so-called ``coincident peak'' as in~\cite{liu2013}. This would be done by appropriately setting the $\alpha_t$ values for the windows when energy prices go up due to the occurrence of a coincident peak.  
\item Our formulation can also be extended to have different delay deadlines ($\tau$) for demands originating in different time slots as required in some settings. E.g., certain data analytic operations might need to be finished by the end of the day creating longer deadlines for jobs arriving earlier in the day and shorter deadlines for late arrivals. 
\end{enumerate}

\vspace{-0.1in}
\section{Stochastic Control}
\label{sec:stochastic}


Whereas OFF can be useful in devising cost-effective workload modulation (as we will corroborate in Section~\ref{sec:eval}), many data center workloads exhibit uncertainty~\cite{rao2011hedging}. Similarly, electricity prices can also exhibit uncertainty - two salient examples are schemes based on spot pricing or schemes that charge higher rates during coincident peaks~\cite{liu2013}. Deviations in demands 
or prices compared to those assumed by OFF can result in poor decision-making (as we also show in Section~\ref{sec:eval}). Consequently, it is desirable to devise workload modulation techniques that can adapt their behavior to workload and price evolution in an {\em online} manner.

Existing literature shows that data center workloads can often be captured well via statistical modeling techniques~\cite{ChenHLNRX08}. Therefore, before devising fully online algorithms in Section~\ref{sec:online}, we first develop an optimization framework based on stochastic dynamic programming (SDP) that leverages predictive models for workload demands and electricity prices.  A demand modulation algorithm resulting from such an SDP is {\em only partly  online}: it can be re-solved whenever the data center updates its predictive models for workloads demands and/or electricity prices to yield the control policy to be used till the next time another adjustment of these models is done. 

\vspace{-0.1in}
\subsection{Stochastic Dynamic Program}
\label{sec:sdp}
We choose as the goal of our SDP the minimization of the {\em expectation} of the sum of the data center's electricity bill and demand modulation penalties over an optimization horizon, and denote it $\bar{\mathcal{O}}(\{\mathcal{A}_t\}, \{\mathcal{D}_t\})$ building upon the notation used in Section~\ref{sec:offline}.   One key difficulty in our formulation arises due to the
somewhat unconventional nature of our SDP wherein our objective is a {\em hybrid} of additive (total energy) and maximal (peak power) components.\footnote{This represents an important difference from other existing SDP formulations in this area, such as~\cite{hegde11}, where their objective is optimizing a purely energy-based cost. } Consequently, we define as our state variable a combination of the residual demands $\{\mathcal{R}_t\}$ (as introduced in Section~\ref{sec:offline}) and the peak power admitted so far. 
Using $y_t$ to denote the peak demand admitted during $[1, ..., t-1]$, we represent the state at time $t$ as the $(|h(t)|+1)$-tuple $s_t=(\{r_{i,t}\}_{i \in h(t)}, y_t)$ or simply $s_t=(\mathcal{R}_t, y_t)$ for brevity. Since $|h(t)|=\tau$ for $\tau < t \le T$, we will simply say that this tuple is of size $\tau+1$.  The peak demand at the end of the optimization horizon $y_{T+1}$ then corresponds to the variable $y_{\rm max}$ employed in OFF.


Since the power demands arriving in different control windows are only known stochastically, we introduce the notation $P_t$ 
to denote these random variables, and use $p_t$ to denote a particular realization of $P_t$.  Correspondingly, we denote by $P_{[t]}=(P_1,...,P_t)$ the history of the demand up to $t$, and by $p_{[t]}=(p_1,...,p_t)$ its particular realization. We assume that the conditional probability distribution of the power demand $P_t$, i.e., ${\sf Pr}\{P_t=p_t | P_{[t-1]}=p_{[t-1]}\}$, is known. Similarly, we use $\alpha_t$ to denote a particular realization of $\Lambda_t$, the random variables for the time-varying prices in different control windows. we denote by $\Lambda_{[t]}=(\Lambda_1,...,\Lambda_t)$ and $\alpha_{[t]}=(\alpha_1,...,\alpha_t)$ the history of the price up to $t$ and its particular realization, respectively. We assume the conditional probability distribution of time-varying price $\Lambda_t$, i.e., ${\sf Pr}\{\Lambda_t=\alpha_t | \Lambda_{[t-1]}=\alpha_{[t-1]}\}$, is known. We also assume that our control/modulation actions do not affect the demand/price arrival process. 
We denote as $a_t^+$ the summation $\sum_{i \in h^+(t)} a_{i,t}$ as before.  The Bellman's optimality rules for our SDP can now be written as follows. For the last control window, we solve the following optimization problem (we denote it as SDP(T)) which gives $V_T(s_T,p_{[T-1]},\alpha_{[T-1]})$: 

\begin{equation*}
\begin{aligned}
&\min_{\mathcal{A}_T,\mathcal{D}_T} \mathbb{E}  \Big\{ \beta y_{T+1} + \alpha_T a_T^+ + l_{\sf drop}(d_{T,T}) \\
                                         & + \sum_{i \in h(T)} l_{\sf delay} (a_{i,T}) ~~| ~~P_{[T-1]}=p_{[T-1]},\Lambda_{[T-1]}=\alpha_{[T-1]} \Big\},
\end{aligned}
\end{equation*}
Subject to:
\begin{equation*}
\begin{aligned}
&p_T - a_{T,T} - d_{T,T} = 0; ~r_{i,T} - a_{i,T} = 0, i \in h(T),\\
&y_{T+1} \ge y_T; ~y_{T+1} \ge a_T^+,\\
&a_{i,T} , r_{i,T} \ge 0, i \in h^+(T); ~y_{T+1} , d_{T,T} \ge 0.
\end{aligned}
\end{equation*}
For control windows $t=T-1, ..., 1$, we solve $SDP(t)$ to obtain $V_t(s_{t},p_{[t-1]},\alpha_{[t-1]})$:
\begin{equation*}
\begin{aligned}
&\min_{\mathcal{A}_t,\mathcal{D}_t} \mathbb{E}  \Big\{\alpha_t a_t^+ + l_{\sf drop} (d_{t,t}) + \sum_{i \in h(t)}l_{\sf delay} (a_{i,t}) \\
                                         & + V_{t+1}(s_{t+1},p_{[t]},\alpha_{[t]}) ~~| ~~P_{[t-1]}=p_{[t-1]},\Lambda_{[t-1]}=\alpha_{[t-1]} \Big\},
\end{aligned}
\end{equation*}
Subject to:
\begin{equation*}
\begin{aligned}
&p_t -  a_{t,t} -d_{t,t} = r_{t,t+1}; ~r_{i,t} - a_{i,t} = r_{i,t+1}, i \in h(t+1), \\
&y_{t+1} \ge y_t; ~y_{t+1} \ge a_t^+,\\
&a_{i,t} \ge 0,  i \in h(t); r_{i,t+1}\ge 0, i \in h(t+1); ~d_{t,t}, y_{t+1} \ge 0.
\end{aligned}
\end{equation*}


\noindent{\bf Computational Tractability Concerns:}
Unfortunately, in its general form, our SDP is likely to be computationally prohibitive due to the well-known ``curse of dimensionality''~\cite{bert}. An inspection of the optimization problems SDP(t) reveals multiple contributors to such scalability limitations: (i) the number of discretization levels $L_p$ and $L_{\alpha}$ used for power demands and energy price, respectively, (ii) the number of control windows $T$ (which depends both on the billing cycle and the control window $\delta$), (iii) the complexity of predicting workload demands and prices (in terms of the amount of historical data needed), (iv) the presence of peak-based charging which results in the incorporation of an additional variable $y_t$ in our state definition, and (v) the option of delaying the workload by up to $\tau$ windows and the non-linear nature of the delay-related penalties (which requires us to maintain $\tau$ residual demands in our state definition). More quantitatively, under stage-independence assumptions on demands and prices, and denoting as $O(R)$ the run-time of sub-problem SDP(t), the runtime for our SDP can be expressed as $O(R\cdot L_p^{2(\tau+2)}\cdot L_{\alpha} \cdot T)$.

\begin{wrapfigure}[11]{r}{0.25\textwidth}
  \centering
  \includegraphics[width=0.23\textwidth]{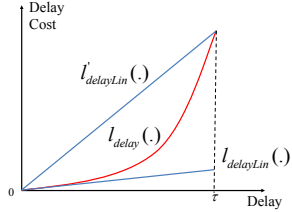}
  \caption{\small Linear approximation of delay cost.} 
\label{fig:linearCost}
\end{wrapfigure}

Of these, 
in our specific problem domain, (i) and (ii) are not scalability barriers (e.g., our choice of $\delta=10$ minutes results in a manageable $T=4320$). Numerous efforts on predicting data center workloads (e.g.,~\cite{Pacheco-SanchezCSMPD11}) suggest the adequacy of (first- or low-order) Markovian models which help ensure that (iii) is also not a scalability barrier.  This leaves (iv) and (v) of which (v) is clearly the more significant - especially for large values of $\tau$ (notice the $L_p^{2(\tau+2)}$ multiplier). 


\vspace{-0.2in}
\subsection{Scalable Approximations of Our {\rm SDP}}
Numerous ideas exist for devising scalable approximations to an SDP~\cite{bert}. The most salient ideas include merging of states and model predictive control.  We investigate two successive simplifications to our SDP for 
overcoming its scalability problems caused by large values of $\tau$. 


{\noindent \bf Linear Approximation of Delay Penalties:} Our first approximation, called SDP$_{\rm Lin}$, employs a linear approximation $l_{\sf delayLin}(.)$ for the function $l_{\sf delay}(.)$ (see Figure~\ref{fig:linearCost} for examples). Specifically, if $l_{\sf delayLin}(x,y)$ has the form $k_{\sf delay}xy$ for $k_{\sf delay}>0$, then it is easily seen that the SDP state can be simplified to the 2-tuple $s_t=(y_t,r_t)$, where $r_t$ denotes the {\em sum} of the residual demands ($\mathcal{R}_t$) at the beginning of $t$ and the optimality rules can be written as follows. 
We denote as $d_t$ and $a_t$ the demand that is dropped and admitted out of $p_t$ respectively, ${a'}_t$ the demand that is admitted out of $r_t$. For the last control window, we solve the following optimization problem (we denote it as SDP$_{\rm Lin}$(T)) which gives $V_T(s_T,p_{[T-1]},\alpha_{[T-1]})$: 

\begin{equation*}
\begin{aligned}
&\min_{a_T,{a'}_T,d_T} \mathbb{E} \Big\{ \beta y_{T+1} + \alpha_T(a_T+a'_T) + l_{\sf drop} (d_T)\\
                                         &  + k_{\sf delay}{a'}_T ~~| ~~P_{[T-1]}=p_{[T-1]},\Lambda_{[T-1]}=\alpha_{[T-1]} \Big\},
\end{aligned}
\end{equation*}
Subject to:
\begin{equation*}
\begin{aligned}
&p_T - a_T - d_T = 0; ~r_T - {a'}_T = 0, \\
&y_{T+1} \ge y_T; ~y_{T+1} \ge a_T+{a'}_T,\\
&a_T, {a'}_T, y_{T+1},d_{T} \ge 0.
\end{aligned}
\end{equation*}
For control windows $t=T-1, ..., 1$, we solve SDP$_{\rm Lin}$(t) to obtain $V_t(s_{t},p_{[t-1]},\alpha_{[t-1]})$:
\begin{equation*}
\begin{aligned}
&\min_{a_t,{a'}_t,d_t} \mathbb{E}\Big\{ \alpha_t(a_t+{a'}_t) + l_{\sf drop}(d_t) + k_{\sf delay}{a'}_t \\
                                         & + V_{t+1}(s_{t+1},p_{[t]},\alpha_{[t]}) ~~| ~~P_{[t-1]}=p_{[t-1]},\Lambda_{[t-1]}=\alpha_{[t-1]} \Big\},
\end{aligned}
\end{equation*}
Subject to:
\begin{equation*}
\begin{aligned}
&(p_t - a_t - d_t) + (r_t - {a'}_t)  = r_{t+1}, \\
&r_t - {a'}_t \ge 0; ~p_t - a_t - d_t  \ge 0, \\
&y_{t+1} \ge y_t; ~y_{t+1} \ge a_t+{a'}_t,\\
& a_t, {a'}_t , d_t, y_{t+1} , r_{t+1}\ge 0.
\end{aligned}
\end{equation*}

SDP$_{\rm Lin}$ offers a significantly reduced runtime of $O(R\cdot L_p^5 \cdot L_{\alpha}\cdot T)$ (again under the assumptions of stage-independence for demands and prices and denoting as $O(R)$ the runtime of a sub-problem SDP$_{\rm Lin}$(t)). Of course, this is at the cost of poorer solution quality, and we empirically evaluate this trade-off in Section~\ref{sec:eval}. 

{\noindent \bf Ignoring the Delay Knob:} Our second approximation, called SDP$_{\rm Drop}$, is based on completely ignoring the (computationally) problematic knob of demand delaying. Notice that it is a special case of SDP$_{\rm Lin}$ in that it can be viewed as employing $\tau=0$. With this approximation, the state simplifies even further to $s_t=y_t$ and the set of control variables shrinks to only $d_t$ to describe dropping. SDP$_{\rm Drop}$ has a reduced runtime of $O(R\cdot L_p^3 \cdot L_{\alpha}\cdot T)$. Furthermore, the structure of the optimal policy for SDP$_{\rm Drop}$ given by the following lemma (proof in the Appendix) provides a way of explicitly finding the optimal control policy~\cite{bert}.




\begin{lem}
\label{lem:SDPoptStructure}
If the demands $p_t$ are independent across $t$,  SDP$_{\rm Drop}$ has a threshold-based optimal control policy. 
Define $\mu_t=\frac{a_t}{p_t},\mu_t \in [0,1]$, $\mathcal{G}_t(\mu_t)=\mathbb{E} \{\alpha_t(\mu_t p_t)\times \mu_t p_t +l_{\sf drop}(p_t-\mu_t p_t)+V_{t+1}(\max\{\mu_t p_t,y_t\})\}$
and $\phi_t(y_t)=\arg \min_{\mu_t \in \mathcal{R}^+}\mathcal{G}_t(\mu_t)$. Then the threshold-based optimal policy $(a_t^*,d_t^*)$ is as follows:\\
\begin{equation*}
\begin{aligned}
(a_t^*,d_t^*)=
\begin{cases}
   (\phi_t p_t, p_t-\phi_t p_t), ~~&{\rm if} \phi_t \le 1\\
    (p_t, 0), &{\rm if} \phi_t > 1
\end{cases}
\end{aligned}
\end{equation*}
\end{lem}

\vspace{-0.2in}
\section{Online Algorithms}
\label{sec:online}

We devise two online control policies that might prove useful 
in scenarios where predictive models for power demands fail.\footnote{Extending these algorithms to also deal with uncertainty in electricity prices is an important direction for future work.}  Due to space constraints, we discuss only one (ON$_{\rm Drop}$) in detail and briefly outline the second one (ON$_{\rm MPC}$). 

\noindent{\bf {\rm ON}$_{\rm Drop}$: Online Algorithm With Only Dropping:} As discussed in Section~\ref{sec:back}, many workloads are delay-intolerant but do allow (possibly implicit) dropping of the power demands they pose (e.g., recall Example 1) with associated loss in revenue. For these scenarios, it is possible - under some assumptions - to devise a completely online algorithm which offers a competitive ratio of $2-\frac{1}{n}$ (i.e., the ratio of its cost and the optimal cost is guaranteed to be upper bounded by $2-\frac{1}{n}$ even under adversarial demands). As two simplifications, we assume (i) $l_{\sf drop}(x)$ has the form $k_{\sf drop}x$ as in Section~\ref{sec:eval}, and (ii) $\alpha_t=\alpha, \forall t$. We leave the extension of ON$_{\rm Drop}$ and its competitive analysis without these assumptions for future work. Notice that the resulting problem setting is identical to that considered for SDP$_{\rm Drop}$ with the difference that whereas ON$_{\rm Drop}$ assumes adversarial power demand inputs, SDP$_{\rm Drop}$ was designed for power demands that can be described stochastically.

The following lemma is key to the design of ON$_{\rm Drop}$.

\begin{lem}
\label{lem:onlydemanddropping}
If only demand dropping is allowed, $l_{drop}(x)=k_{\sf drop}x$, and $\alpha_t=\alpha$, then the optimal demand dropping threshold $\theta$ has the following form: if we denote $\hat{p}_t$ as the $t^{\rm th}$ largest demand value in $\{p_t\}_{t=1}^{T}$, then $\theta=\hat{p}_n$, where $n=\lceil \frac{\beta}{(k_{\sf drop}-\alpha)} \rceil$.
\end{lem}
Lemma~\ref{lem:onlydemanddropping} tells that the optimal demand dropping threshold of OFF when only dropping is allowed equals to the $n^{\rm th}$ largest value in $\{p_t\}_{t=1}^{T}$. If we can keep track of the $n^{\rm th}$ largest power value in an online fashion, finally we can find the optimal demand dropping threshold after observing all the demand in the optimization horizon. So we exploit Lemma~\ref{lem:onlydemanddropping}, to devise ON$_{\rm Drop}$ as follows:

\begin{enumerate}
\item Initialization: Set demand dropping threshold $\theta=0$, $n=\lceil \frac{\beta}{(k_{\sf drop}-\alpha)} \rceil$.
\item At time $t$, sort $p_1,...,p_t$ into $\hat{p}_1,...,\hat{p}_t$ such that $\hat{p}_1 \geq \hat{p}_2 \geq ... \geq \hat{p}_t$.
\item Update $\theta$ as follows: If $t<n$, $\theta=0$; Otherwise, $\theta=\hat{p}_n$.
\item Decision-making: Admit $\min(p_t,\theta)$, drop $[p_t-\theta]^+$. 
\end{enumerate}

\begin{thm}
\label{thm:cr}
{\rm ON}$_{\rm Drop}$ offers a competitive ratio of $2-\frac{1}{n}$.
\end{thm}

We present the proof in the Appendix. Based on an induction on the length of the optimization horizon $T$, we find an upper bound of the difference between the total admitted demand of the optimal solution and ON$_{\rm Drop}$ in our proof. 
One useful way of understanding ON$_{\rm Drop}$ is to view it as a generalization of the classic ski-rental problem~\cite{Karlin1990}: increasing the dropping threshold is analogous to purchasing skis while admitting/dropping demands according to the existing threshold is analogous to renting them. The generalization lies repeating the basic ski-rental-like procedure after every $n$ slots.

\noindent{\bf {\rm ON}$_{\rm MPC}$: Model Predictive Control:}
Finally, we also develop an online algorithm by adapting ideas from Model Predictive Control (MPC)~\cite{macie}, a well-known suboptimal control theory. Adapting ideas from MPC, we assume that at time $t$, our control has access to accurate/predictable information about the inputs (i.e., the power demand/electricity price ) during a short-term future time window  $[t, t+h-1]$ of size $h$ (called the ``lookahead window''). It then solves a smaller version of the original optimization problem defined over an optimization window of size $H$ (called the ``rolling horizon'') where $h \le H \le T$, assuming that the power demand during $[t+h, t+H]$ to be given by the mean power demand (or time-of-day demand\footnote{Many real-world workloads of data centers exhibit strong time-of-day behavior, as shown in Section~\ref{sec:eval}, which can serve as the base demand for prediction model. Note that in this paper our focus is exploring demand modulation under different pricing schemes instead of demand prediction.}) obtained from previous observations. More sophisticated versions of MPC have been employed in recent related work~\cite{WangC08}. Our purpose in devising this algorithm is to employ it as a baseline that we empirically compare with our other algorithms. We evaluate ON$_{\rm MPC}$ in Section~\ref{sec:eval} with parameter settings in Table~\ref{tbl:params}.

{\bf Algorithmic Details:} At time slot $t$, we solve the following convex program $V_t(s_t,p_{[t-1]})$ in the optimization window $[t,t+H-1]$:
\begin{equation*}
\begin{aligned}
\min_{\mathcal{A}_m,\mathcal{D}_m} & \beta y_{\rm t+H} + \sum_m \Big\{\alpha_m a_m^+  + \sum_{i \in h^+(m)} l_{\sf drop}(d_{i,m}) \\
&+  \sum_{i \in h(m)} l_{\sf delay}(a_{i,m},m-i)\Big\} \\
\textrm{subject to}~&(\ref{eq:eqone})-(\ref{eq:eqlast})~\textrm{with $m$ replacing}~t, \forall m \in [t,t+H-1].
\end{aligned}
\end{equation*}

Following this, the optimal control at the current control window $t$, i.e., $(a_{i,t}^*,d_{i,t}^*)$, $i\in h(t)^+$ is implemented while the rest of the optimal control sequence for $m \in \{t+1, ..., t+H-1\}$ is discarded, and we shift the optimization horizon forward by one time slot (control window), and repeat the above procedure till $t+H-1=T$.

\vspace{-0.1in}
\section{Experimental Results}
\label{sec:eval}

We evaluate our algorithms using both power demands derived from real-world traces and constructed synthetically. {\em Our metric is "cost savings" which is the percentage reduction in costs due to a particular algorithm compared to doing no demand modulation.} We assume $l_{\sf drop}(x)=k_{\sf drop}x$ and $l_{\sf delay}(x,t)=k_{\sf delay}t^2 x$, respectively (here we choose linear dropping cost which is similar to the cost model used in ~\cite{freris2010resource}. Also note that our choice of a quadratic relation between delay-related costs and the amount of delay has been found an appropriate choice in some studies~\cite{van1995dynamic}; other recent studies have assumed linear costs in their evaluation~\cite{LiuLWLA11}). 

\vspace{-0.1in}
\subsection{Parameters and Workloads}
\begin{table}[htbp]
\scriptsize
\centering
\begin{tabular}{||l|l|l||}
\hline \hline
Param.                & Description                             & Value \\
\hline \hline
$\alpha, \beta$       & Energy price (\$/kWh),                 &  0.046, \\
                      & Peak power price (\$/kW/month)         &  17.75~\cite{sceg} \\ \hline
$k_{\sf drop}$         & Cost incurred by                        & 0.72 \\
                      & dropping demand (\$/kWh)                & \\\hline
$k_{\sf delay}$        & Cost incurred per                       & 0.02~\cite{LiuLWLA11}\\
                      & $\delta$ for delaying (\$/kW/month))    & \\\hline
$\tau$                & Maximum delay allowed                   & 6$\delta$ (1 hour)\\\hline
$\delta$              & Control window                          & 10 minutes\\\hline
$T$                   & Optimization window                     & 4320 (30 days / $\delta$)\\\hline
$H, h$                & Receding horizon,                       & 1 day\\
                      & Lookahead window                        & 6 hours \\\hline
\hline
\end{tabular}
\caption{\small Various problem parameters.}
\label{tbl:params}
\end{table}

Table~\ref{tbl:params} lists the values used for various parameters in our evaluation, along with sources when applicable. We set $k_{\sf drop}=k_{\sf delay}*\tau^2$ since: (i) dropping can be thought of as an extreme case of delaying (i.e, for an infinite amount of time) in our formulations, and (i) when some demand unit has been delayed by $\tau$, the only knob for additional demand modulation that remains is dropping it. 
We employ three real-world power demands presented in recent studies: Google, Facebook, MediaServer (streaming media)~\cite{LiKXLSY08,FanWB07,ChenGGK11}. Additionally, we create a synthetic power demand series with an emphasis on including an unpredictable surge in power demand (e.g., as might occur due to a flash crowd). Each power demand series spans 30 days (a typical electric utility billing cycle\footnote{~\cite{XuL2013} evaluates the proposed algorithm with planning period equal to one day under peak-based tariff, which is unrealistic considering the billing cycle of real-world tariff. }), with each point in the series corresponding to the average power demand over a 10 minute period. Synthetic is built by adding a high power surge to the demand for MediaServer  on the 15$^{\rm th}$ day. Since in many cases the real-world workload of data centers exhibits strong time-of-day behavior, we use this time-of-day plus a zero-mean Gaussian noise in our stochastic control models SDP$_{\rm Lin}$ and SDP$_{\rm Drop}$. We show these power demands and their time-of-day behaviors in Figure~\ref{fig:workloads}. The peak demand is $3MW$ and $5.022MW$ for the three real-world workloads and Synthetic, respectively.

To help understand the features of our workloads, we select two workload properties that we intuitively expect to be informative about achievable cost savings: 
(i) the {\em peak-to-average ratio} (PAR), which captures the peak-shaving potential of the workload and (ii) {\em peak width} (P$_{70}$), which is defined as the percentage of the number time slots in which the power demand value is larger than 70\% of the peak power demand. We show these parameters in Fig.~\ref{fig:workloads}. 
To give readers an intuitive feeling of the statistical features of the workloads (complementary to PAR and P$_{70}$), we show the histograms of these workloads in Figure~\ref{fig:workloadHist}. Higher frequency of occurrence of large power values implies larger P$_{70}$ (e.g., Google), and a "thinner" histogram may imply higher PAR (e.g., Synthetic). 

\begin{figure}[htbp]
\begin{center}
\includegraphics[width=0.23\textwidth]{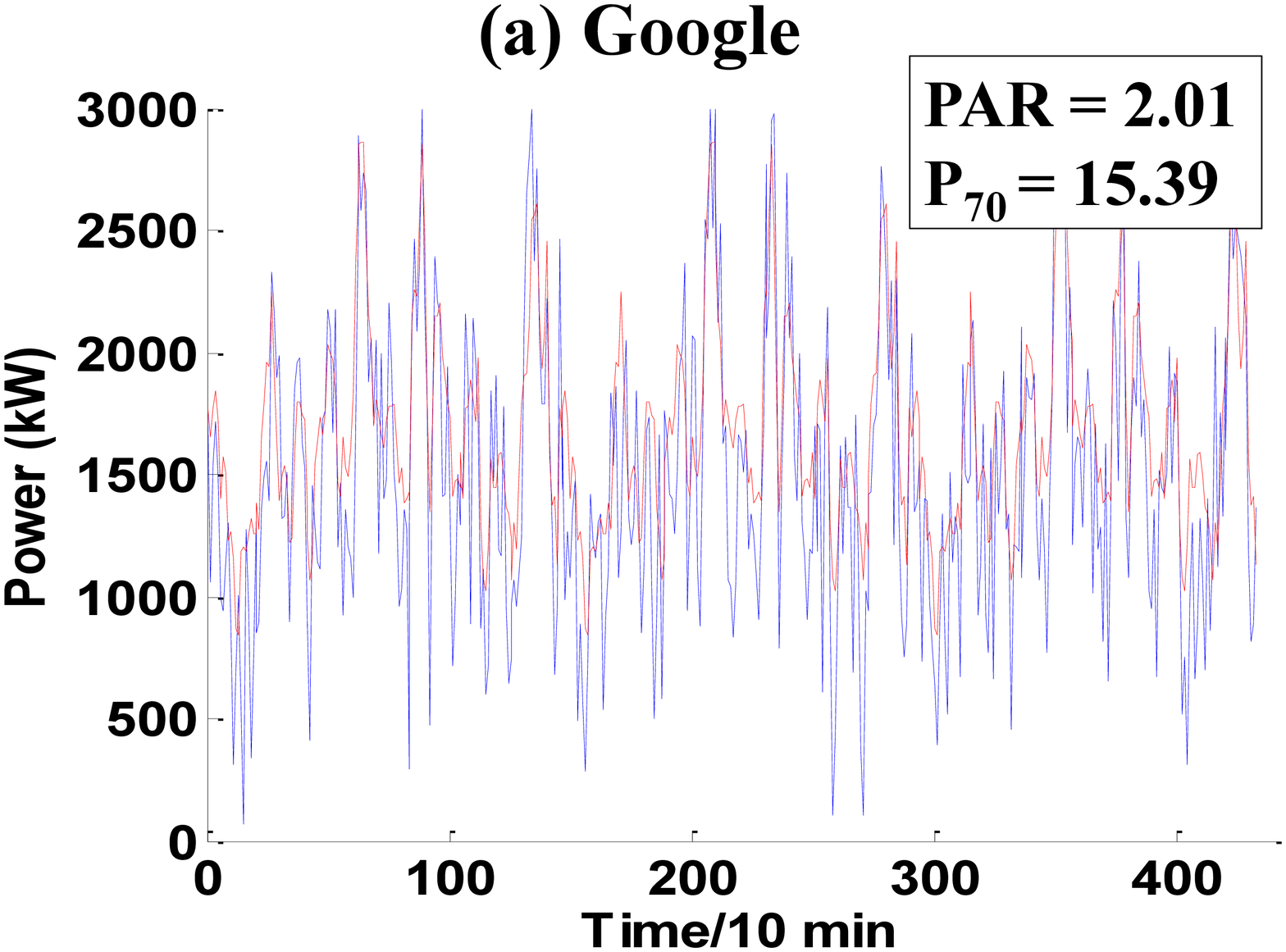}
\includegraphics[width=0.23\textwidth]{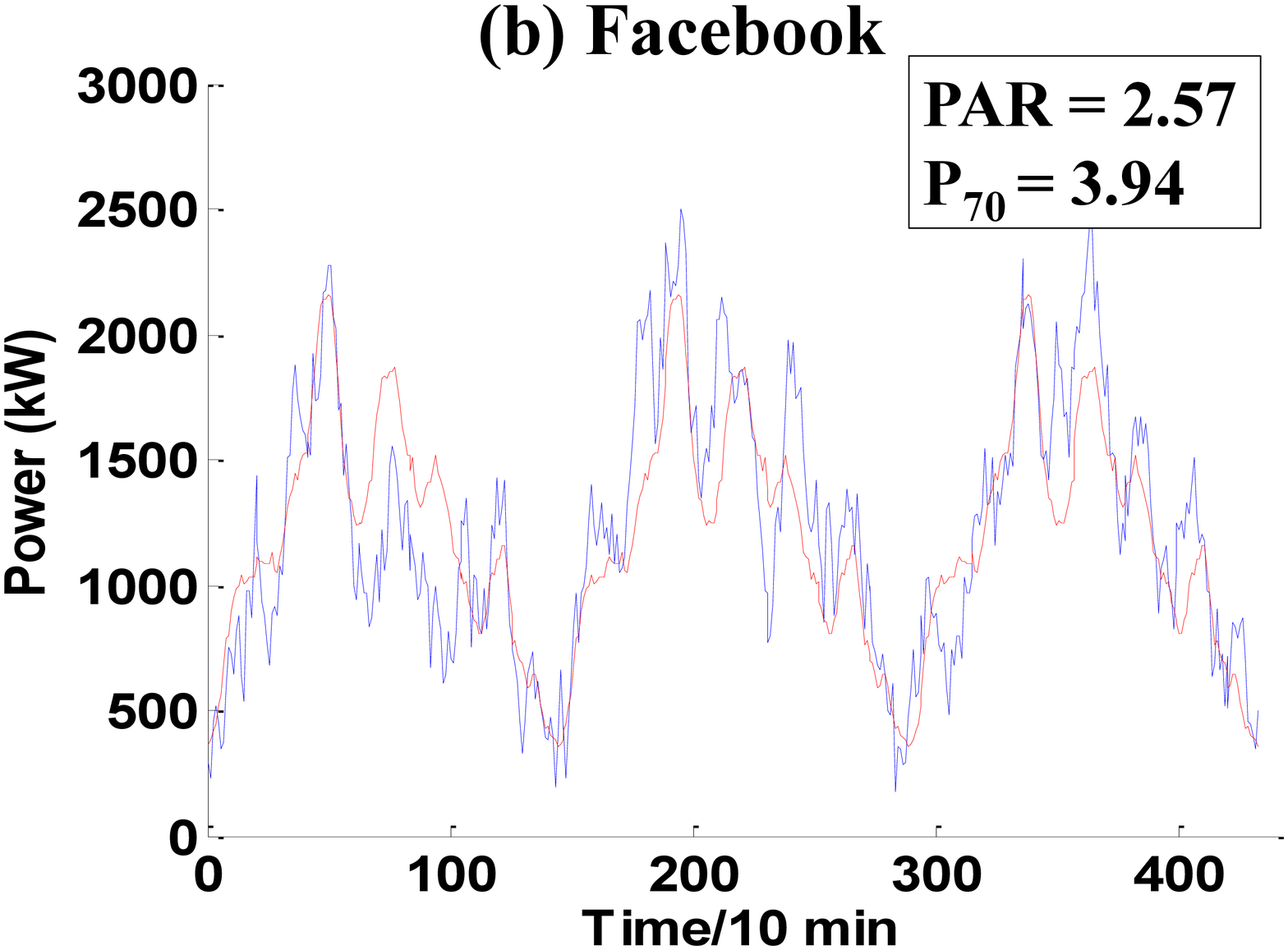}\\
\includegraphics[width=0.23\textwidth]{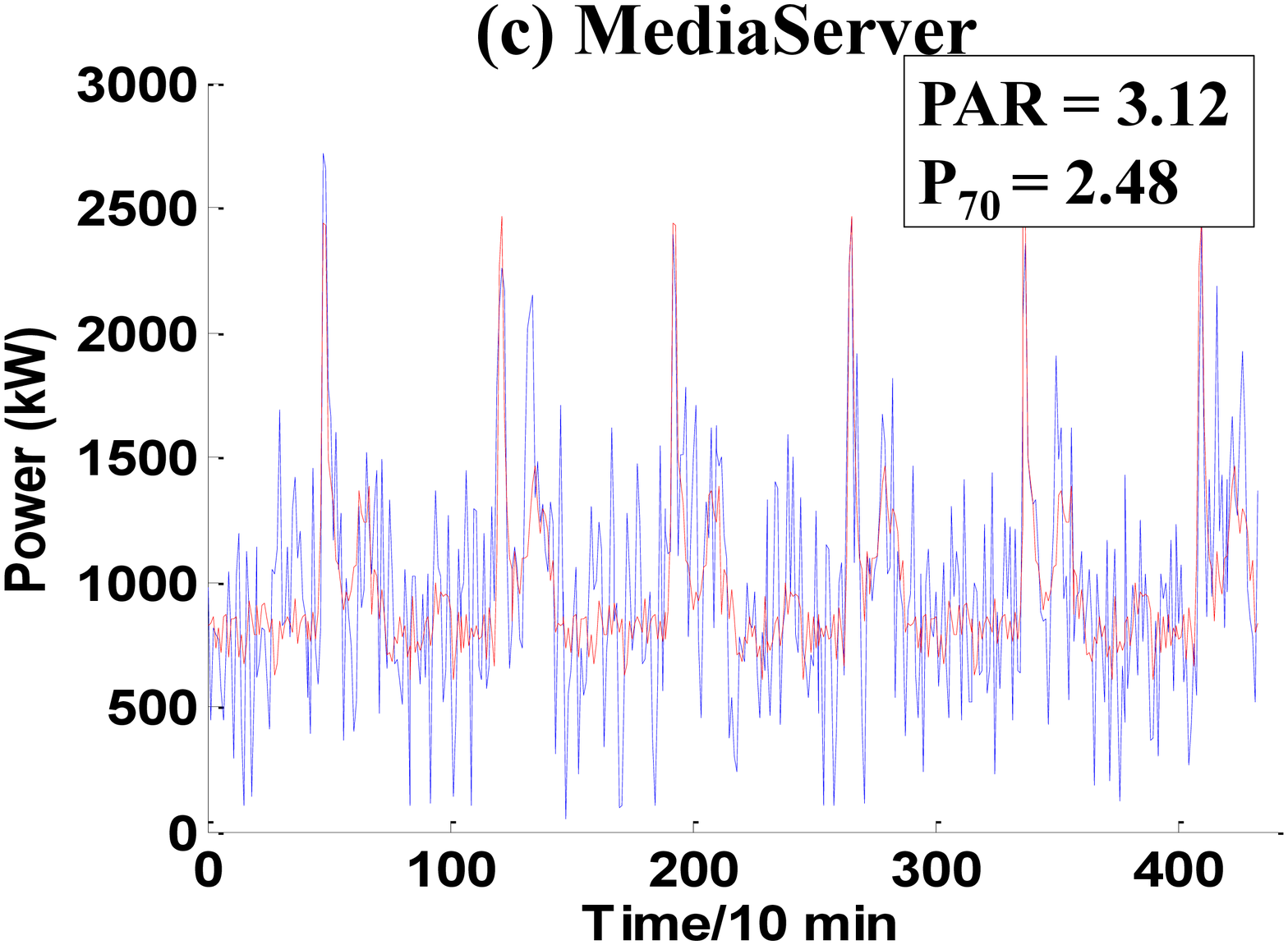}
\includegraphics[width=0.23\textwidth]{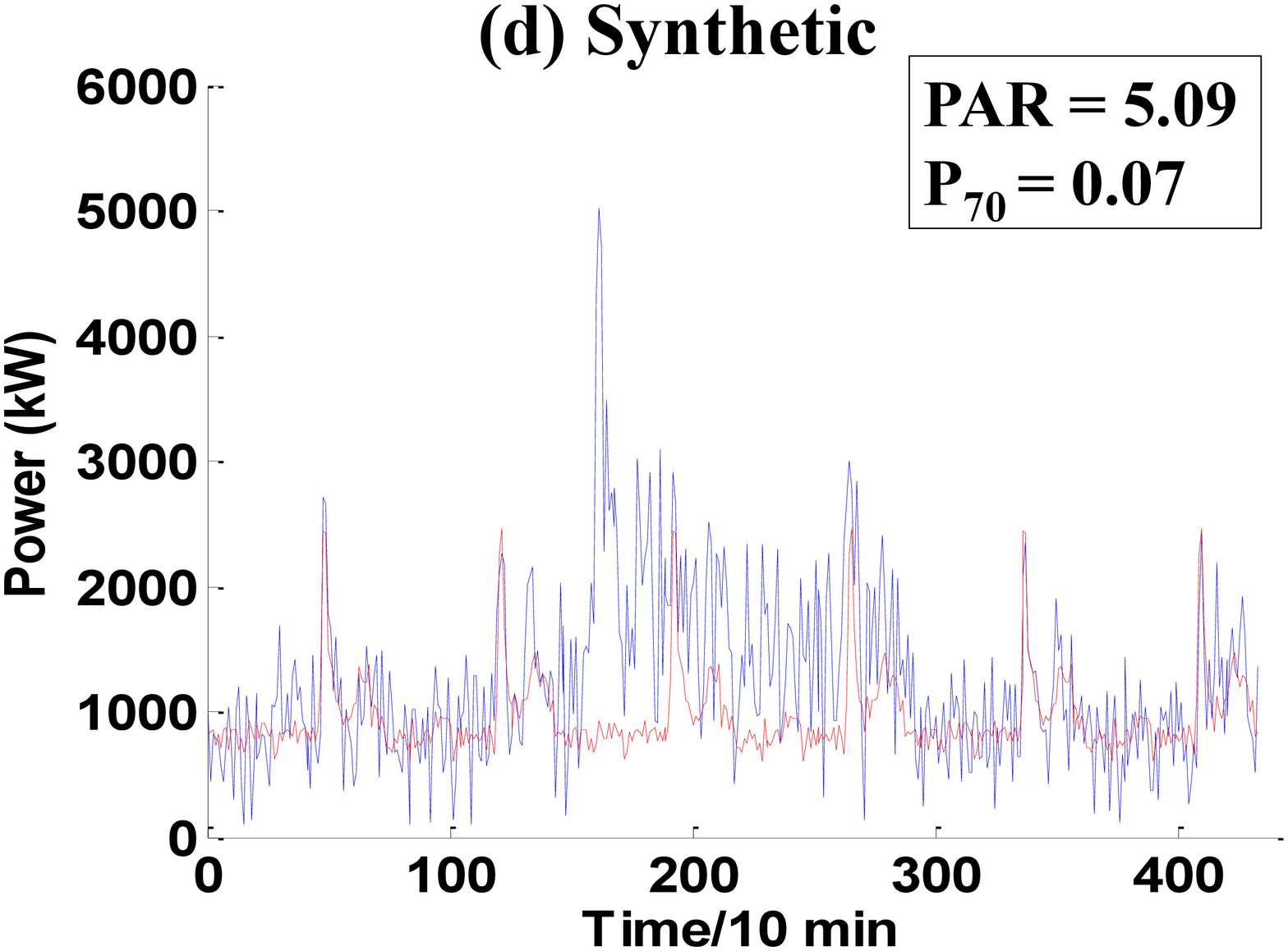}
\end{center}
\caption{\small Workloads from the 14th day to the 16th day. The blue line is the original power demand, and the red line is the time-of-day behavior of the demand.}
\label{fig:workloads}
\end{figure}
\begin{figure}[htbp]
\begin{center}
\includegraphics[width=0.23\textwidth]{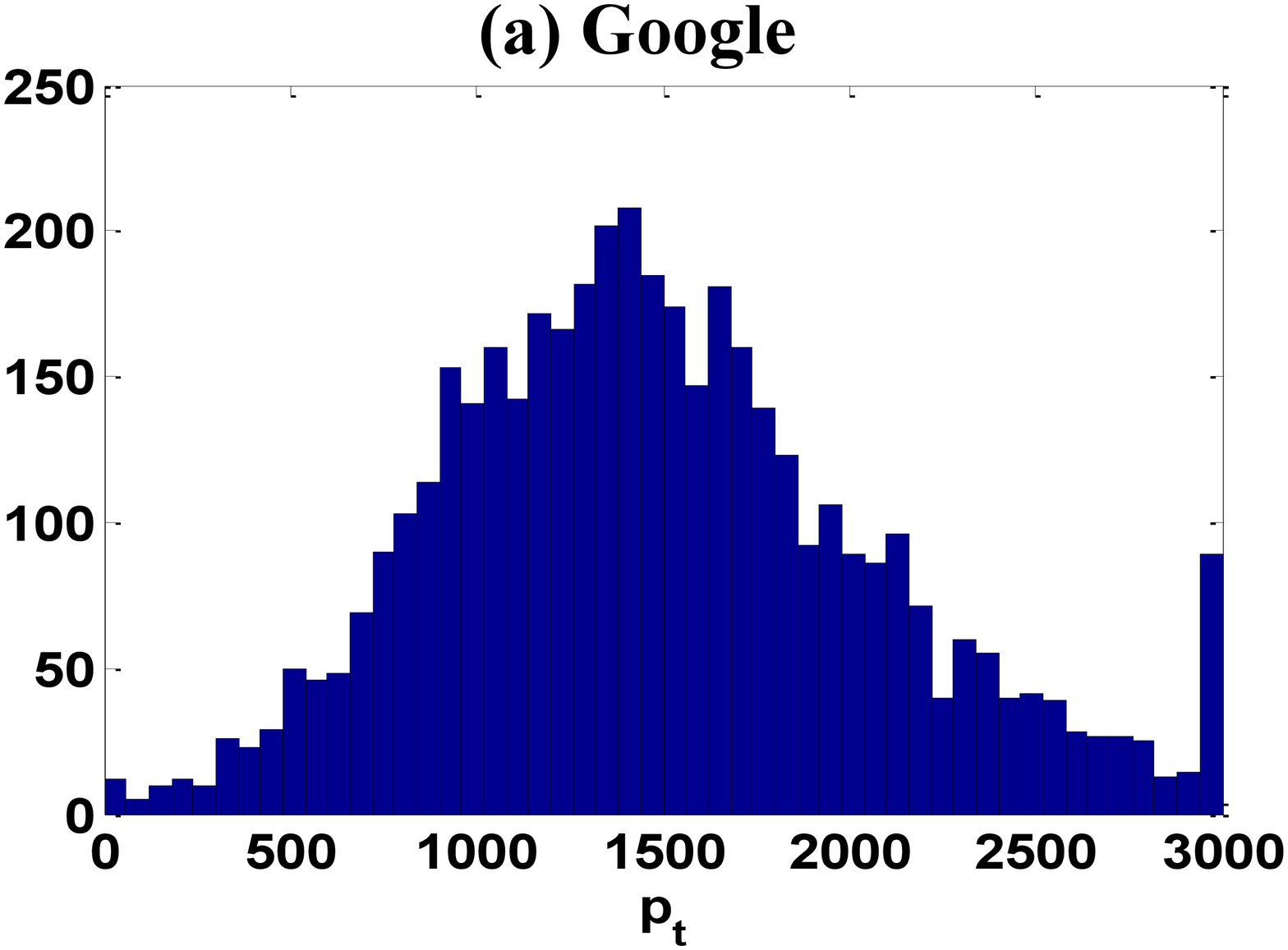}
\includegraphics[width=0.23\textwidth]{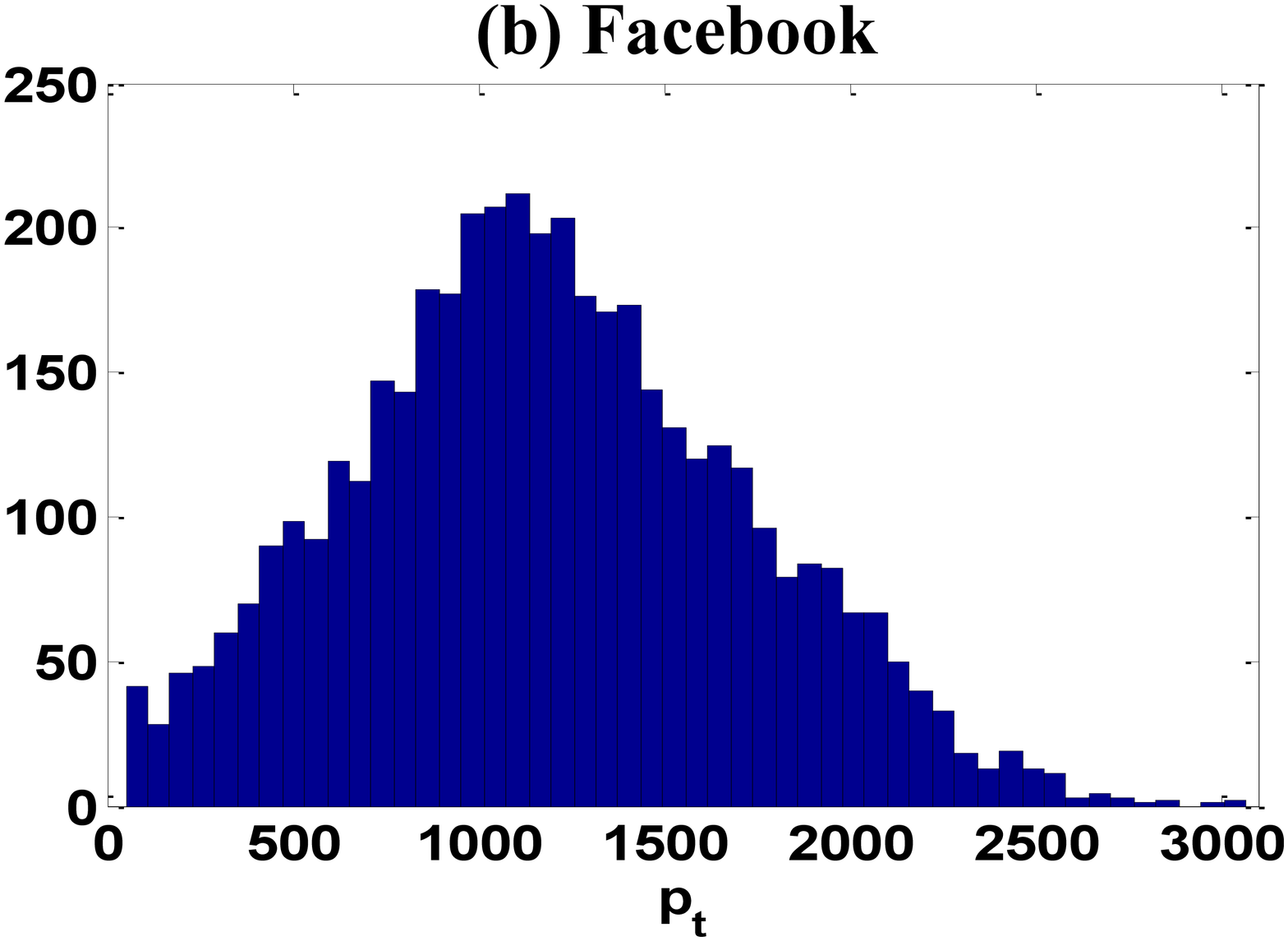}\\
\includegraphics[width=0.23\textwidth]{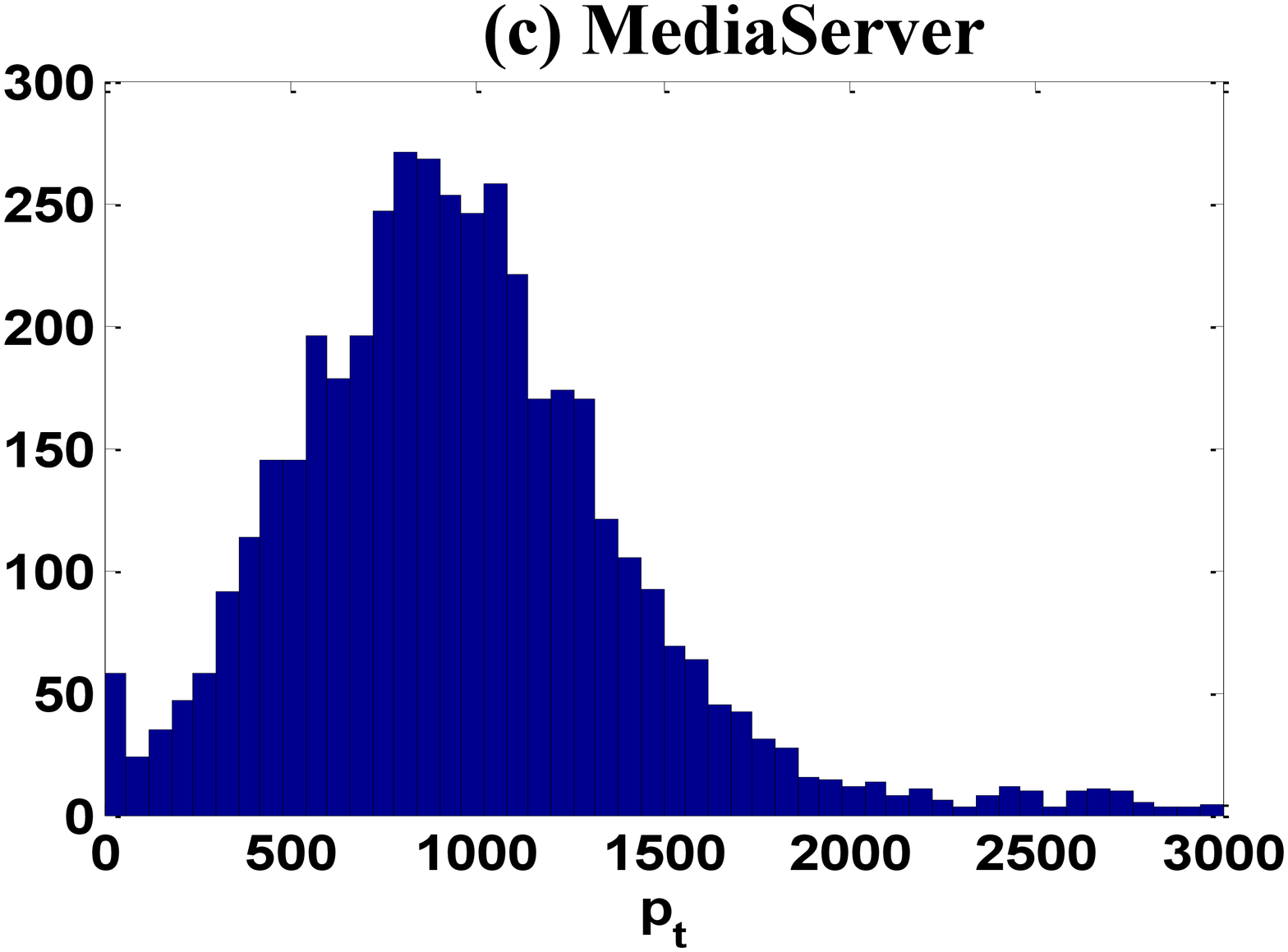}
\includegraphics[width=0.23\textwidth]{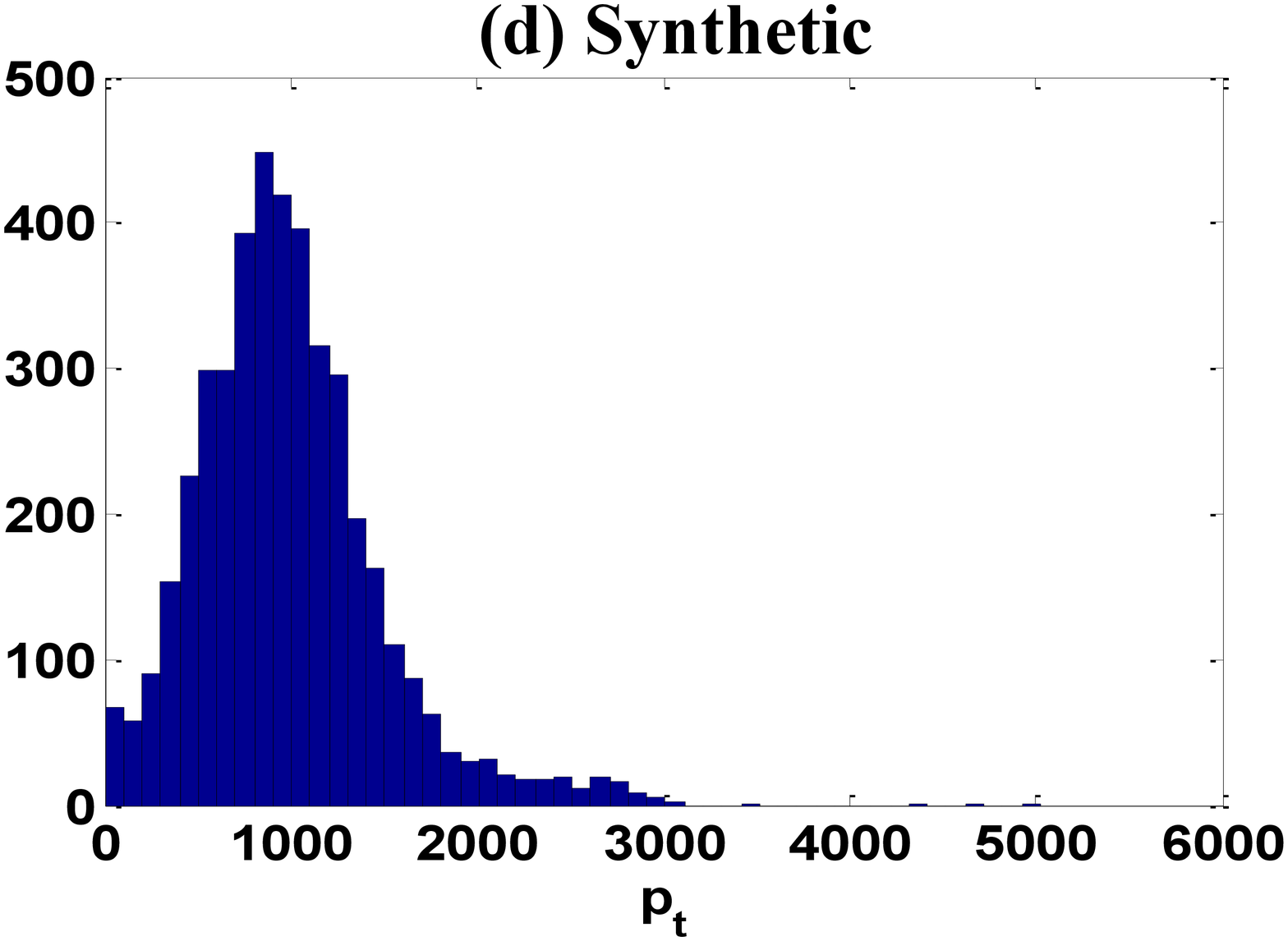}
\end{center}
\caption{Histograms of the power demands within an optimization window for our workloads.}
\label{fig:workloadHist}
\end{figure}

\vspace{-0.1in}
\subsection{Peak Pricing Schemes}
\label{sec:evalPeak}

\begin{table*}[htbp]
\scriptsize
\centering
\begin{tabular}{||l|l|l|l|l|l|l|l|l|l|l|l||}
\hline \hline

Control knobs &\multicolumn{4}{|c|}{Only dropping}              & \multicolumn{3}{|c|}{Only delaying}& \multicolumn{3}{|c|}{Dropping+Delaying} \\
\hline \hline
Workload      & OFF & ON$_{\rm MPC}$ & SDP$_{\rm Drop}$ & ON$_{\rm Drop}$ & OFF & ON$_{\rm MPC}$  & SDP$_{\rm Lin}$  & OFF & ON$_{\rm MPC}$ &SDP$_{\rm Lin}$\\ \hline
Google	      &  1.45 & -0.21      &  1.14        &  -0.29      & 10.74 & 9.92        &  2.83       & 10.74 & 9.92       &  2.83      \\ \hline
Facebook      & 13.11 & 9.26       &  12.52       &   8.95      & 16.39 & 9.96        &  11.79      & 16.97 & 9.96       &  11.79     \\ \hline
MediaServer   & 13.82 & 11.21      &  8.13        &  11.66      & 27.64 & 27.16       &  18.26      & 27.64 & 27.16      &  18.26     \\ \hline
Synthetic     & 36.56 & 30.80      &  29.86       &  32.49      & 38.42 & 32.40       &  -10.37     & 45.05 & 32.43      &   -10.37    \\ \hline
\hline
\end{tabular}
\caption{\small Cost savings ($\%$) offered by different algorithms under peak-based tariff.}
\label{tbl:peakCostSaving}
\end{table*}
\begin{figure}[htbp]
  \centering
  \includegraphics[width=0.35\textwidth]{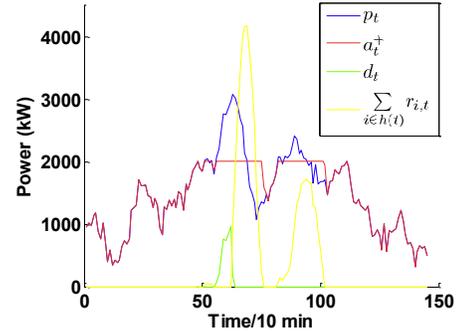}
  \caption{\small Demand modulation by OFF on Facebook (22nd day).}
\label{fig:peakModulation}
\end{figure}
Figure~\ref{fig:peakModulation} shows the results of demand modulation by OFF on Facebook (22nd day). In this example, we find that a large amount of power demand has been deferred (see yellow curve; $\sum_{i\in h(t)}r_{i,t}$ is the aggregate deferred demand from up to $\tau$ time slots before $t$) to the power "valley" right after the peak. However, dropping (green line) is still used to further shave the peak since we can only postpone the demand up to $\tau$ time slots (1 hour in our setting).

Table~\ref{tbl:peakCostSaving} presents the cost savings ($\%$) offered by different algorithms under the peak-based pricing scheme with the parameters ($\alpha, \beta$) shown in Table~\ref{tbl:params}. 
Let us first consider using only delaying. 
Deferring demand is possible only when a near-peak demand $p_t$ is followed immediately (within 1 hour) by much lower power demands, and the costs saving due to the resulting peak reduction is larger than the delay cost incurred. 
The lower and longer these succeeding low demand periods are (typically for larger PAR and smaller P$_{70}$), the better should be the cost savings achieved. Our experiment results verify these intuitions. As shown in Figure~\ref{fig:workloads}, PAR increases drastically from Google, Facebook, MediaServer, Synthetic (in that order) while P$_{70}$ decreases, and correspondingly, the cost savings improve (from $10.74\%$ for Google to $38.42\%$ for Synthetic).

Next, let us consider only using dropping. According to Lemma~\ref{lem:onlydemanddropping}, the optimal demand dropping threshold is determined by $\hat{p}_n, n=\lceil \frac{\beta}{(k_{\sf drop}-\alpha)} \rceil$ in the non-increasing array of $\{\hat{p}_t\}_{t=1}^{T}$. Since $n$ is fixed for all workloads given the pricing parameter settings, a larger P$_{70}$ implies higher probability of the optimal dropping threshold ($\theta$) being high, which means less power demand can be dropped to reduce the peak demand with lower cost savings. Among our workloads, Google has very ``wide'' peak/near-peak power values, which is the reason for the meager 1.45\% cost savings. On the other hand, Synthetic exhibits ``sharp'' and ``tall'' peaks, for which a greater cost saving (36.56\%) is possible.

Thirdly, let us consider the cost savings when both knobs are allowed: why are the cost savings very similar to those when only delaying is allowed? If a delayed unit of demand can be serviced with a small delay, the cost incurred is small compared to that for dropping. The low demand periods within 1 hour of near-peak demands in our workloads are plentiful, making dropping a rarely used knob. Only for Synthetic do we find that 
dropping demand helps improve cost saving by a non-trivial amount: from 38.42\% with only delaying to 45.05\% with both the knobs.

Finally, let us consider how our online/stochastic control algorithms perform.  ON$_{\rm MPC}$ performs well for almost all workload/control-knob combinations, except the $-0.21\%$ for Google with only dropping. Note that, in this case, there isn't much room for savings to begin with (the optimal savings are only 1.45\%. We report similar observations for ON$_{\rm Drop}$ with -0.29\% cost saving for Google (note that these are in line with the competitive ratio we found for ON$_{\rm Drop}$). SDP$_{\rm Drop}$ works well under all workloads, since the workloads exhibit strong time-of-day behavior, implying our SDP has a reliable prediction model to work with. However, we observe -10.38\% cost saving from SDP$_{\rm Delay}$ for Synthetic since our model is unable to capture its flash crowd.

{\bf Key Insights:} (i) For our workloads and parameters, delaying demand offers more benefits more than dropping, (ii) workload properties strongly affect the cost savings under peak-based tariff, and (iii) our stochastic control techniques are able to leverage workload prediction to offer near-optimal cost savings when there is no unexpected flash crowd. 
\vspace{-0.1in}
\subsection{Time-varying Prices}
\label{sec:evalTimeVarying}
Many prior papers have explored various ways to reduce operational costs of data center(s) under time-varying pricing tariff~\cite{urgaonkar-sigmetrics2011,deng2013multigreen,yao2012data} or geographically varying energy price~\cite{LiuLWLA11,RenHX12}. Since our problem formulation is general enough to incorporate time-varying pricing tariff, we also do experiments to find insights for demand modulation under such schemes.

\begin{wrapfigure}[13]{r}{0.25\textwidth}
  \centering
  \includegraphics[width=0.23\textwidth]{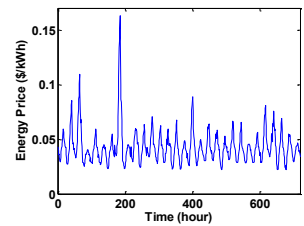}
  \caption{\small Hourly electricity price from 07/01/2012 to 07/30/2012 in zone A of National Grid U.S.~\cite{nationalgrid}}
\label{fig:price}
\end{wrapfigure}


We restrict our attention to scenarios where dropping cost is always larger than energy cost ($l_{\sf drop}(x)>\alpha_t x,\forall t$), the only control knob we can choose is deferring demand. Note that other papers have considered coincident peak pricing where this no longer holds and dropping may also be employed for cost-efficacy~\cite{liu2013}. Throughout this section, we set $k_{\sf delay}=0.01\$/kWh$ as in~\cite{LiuLWLA11}, and evaluate cost savings both for linear delay costs ($l_{\sf delay}(x,t)=k_{\sf delay}tx$) and quadratic delay costs ($l_{\sf delay}(x,t)=k_{\sf delay}t^2 x$). 
Figure~\ref{fig:price} shows the time-varying price time-series we use in our experiments, which is an example of charge imposed on commercial or industrial customers ($>2MW$) according to National Grid U.S.~\cite{nationalgrid}. The energy price $\alpha_t$ ranges from $0.022\$/kWh$ to $0.163\$/kWh$.

\begin{table}[htbp]
\scriptsize
\centering
\begin{tabular}{||l|l|l||}
\hline \hline

Control knob &\multicolumn{2}{|c|}{Only delaying}\footnote{Lots of research has been done on cost minimization under time-varying price. Here we only do experiments by OFF since our focus is on exploring cost saving potential of demand modulation under different pricing schemes.} \\ \hline \hline
Workload      &Linear delay cost & Quadratic delay cost \\ \hline
Google	      &  9.21            & 7.97          \\ \hline
Facebook      &  3.02            & 1.55             \\ \hline
MediaServer   &  7.95            & 6.45            \\ \hline
Synthetic     &  7.77            & 6.34              \\ \hline
\hline
\end{tabular}
\caption{\small Cost savings ($\%$) of different delay cost model under time-varying pricing tariff.}
\label{tbl:timeVaryCostSaving}
\end{table}

Our observations are presented in Table~\ref{tbl:timeVaryCostSaving}, and they may be summarized as follows:(i) the workload properties that we found affecting cost savings significantly for peak-based pricing appear to have a less-easy-to-explain influence on cost savings for time-varying prices. Although Google, MediaServer, and Synthetic have very different features (e.g., PAR and P$_{70}$), the cost savings achieved for them under time-varying pricing are  similar. On the other hand, Facebook demand experiences a much lower cost saving. 
(ii) Workloads that are more sensitive/less tolerable to delay, (with quadratic delay cost in our case), tend to gain less benefit out of demand modulation. 

Next, we explore the impact of the correlation between power demands and energy prices on cost savings. Intuitively, the more the two are correlated, the better cost savings we would expect. 
Our experimental results verify this intuition. When demands and prices are positively correlated (Figure~\ref{fig:PriceWorkloadCorr}(a)), we set $\alpha_t=\frac{p_t-p_{\rm min}}{p_{\rm max}-p_{\rm min}}(\alpha_{\rm max}-\alpha_{\rm min})+\alpha_{\rm min}$, and the cost saving is 11.98\%; when they negatively correlated (Figure~\ref{fig:PriceWorkloadCorr}(b)), we set $\alpha_t=\frac{p_{max}-p_t}{p_{max}-p_{min}}(\alpha_{max}-\alpha_{min})+\alpha_{min}$, and the cost saving is only 3.41\%. In both experiments, $\alpha_{\rm min}=0.022\$/kWh$, $\alpha_{\rm max}=0.183\$/kWh$.

\begin{figure}[htbp]
\begin{center}
\includegraphics[width=0.23\textwidth]{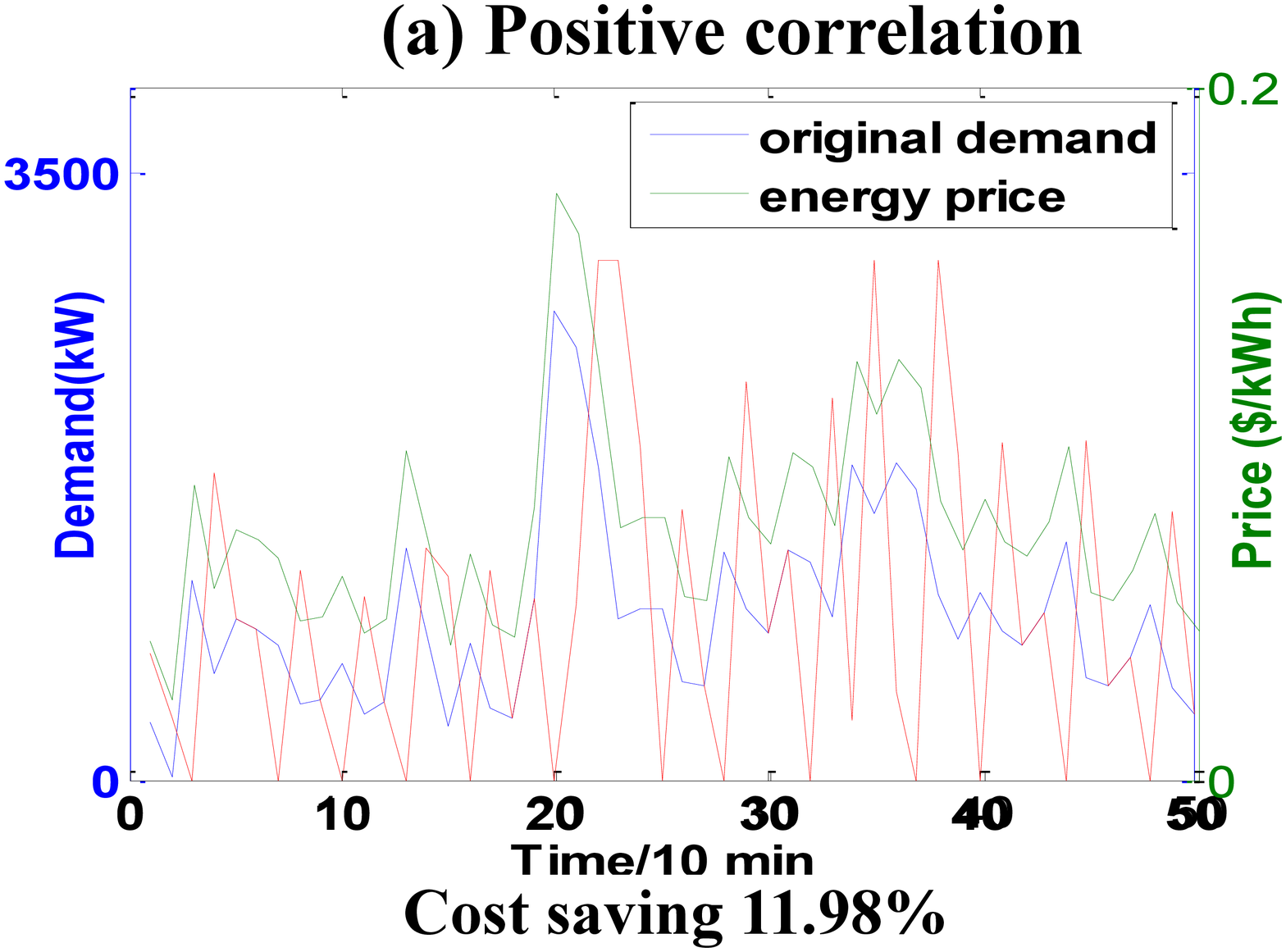}
\includegraphics[width=0.23\textwidth]{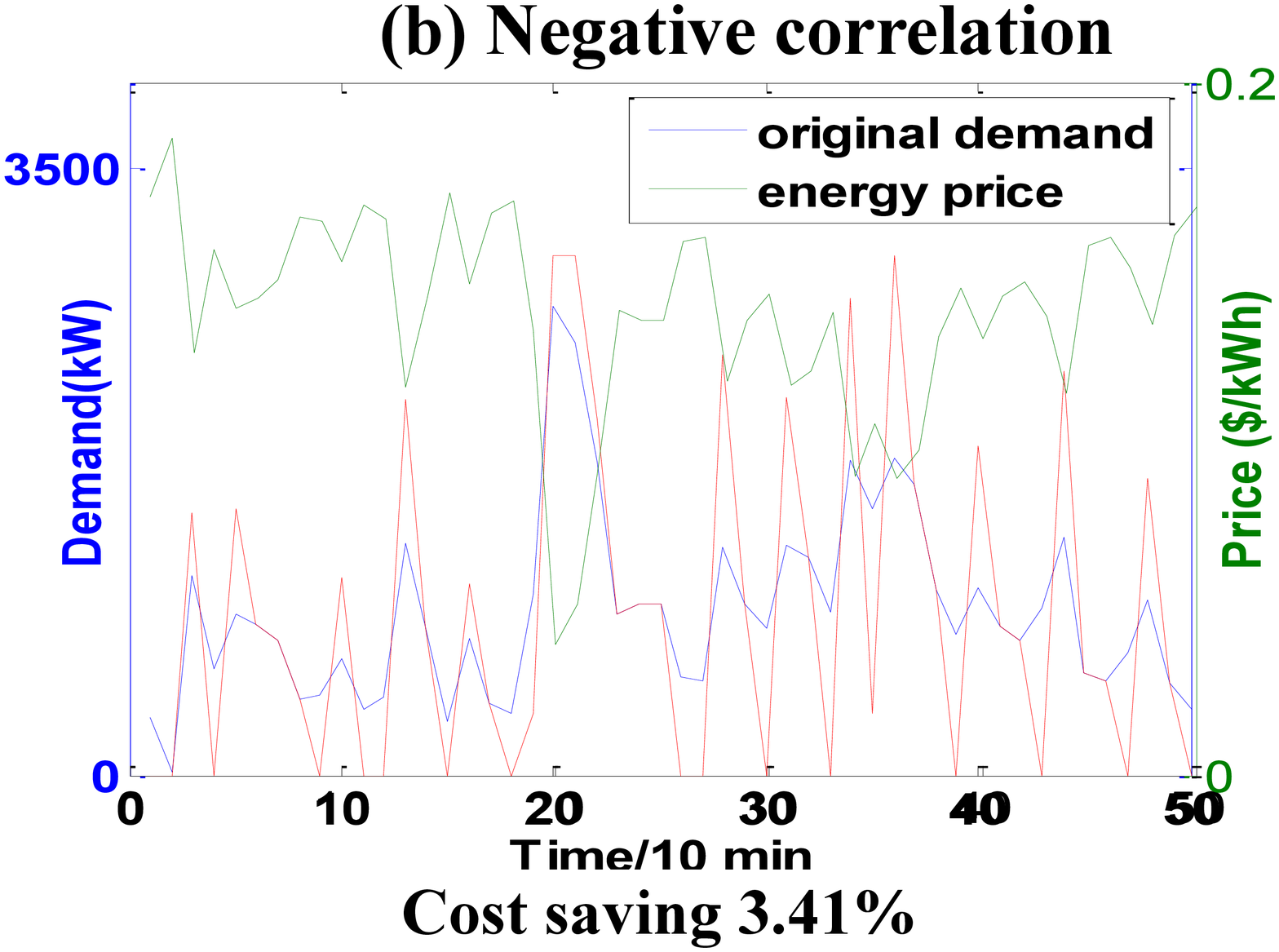}
\end{center}
\caption{\small Impact of correlation between workload and price. Red line is the modulated demand.}
\label{fig:PriceWorkloadCorr}
\end{figure}

{\bf Key Insights:} (i) workload properties (PAR, P$_{70}$, etc. appear to have a lower (or less clear) impact on cost savings compared to peak-based pricing, 
(ii) cost savings greatly depend on the energy price fluctuation and the actual delay penalty (linear vs. quadratic), and   (iii) delaying is more effective as a knob for time-varying prices when demands and prices are positively correlated.

\section{Conclusions and Future Directions}
\label{sec:conclus}


We formulated optimization problems to study how data centers might modulate their power demands for cost-effective operation given three key complex features exhibited by real-world electricity pricing schemes: (i) time-varying prices (e.g., time-of-day pricing, spot pricing, or higher energy prices during "coincident" peaks) and (ii) separate charge for peak power consumption. Our focus was on demand modulation at the granularity of an entire data center or a large part of it, and our work was complementary to a significant body of emergent work in this space (including research threads that have looked at supply-side techniques and demand-side techniques based on additional energy generation or storage sources). For computational tractability reasons, we worked with a fluid model for power demands which we imagined could be modulated using two abstract knobs of demand dropping and demand delaying (each with its associated penalties or costs). For data centers with predictable workloads, we devised a stochastic dynamic program (SDP) that could leverage such predictive models. We also devised approximations (SDP$_{\rm Lin}$ and SDP$_{\rm Drop}$) for our SDP that might be useful when the SDP is computationally infeasible. 
We also devise fully online algorithms (ON$_{\rm Drop}$ and ON$_{\rm MPC}$) that might be useful for scenarios with poor power demand or utility price predictability. For ON$_{\rm Drop}$, we proved a competitive ratio of $2-\frac{1}{n}$. Finally, using empirical evaluation with both real-world and synthetic power demands and real-world prices, we demonstrated the efficacy of our techniques: (i) demand delaying is more effective than demand dropping regarding to peak shaving (e.g., 10.74\% cost saving with only delaying vs. 1.45\% with only dropping for Google workload) and (ii) workloads tend to have different cost saving potential under various electricity tariffs (e.g., 16.97\% cost saving under peak-based tariff vs. 1.55\% under time-varying pricing tariff for Facebook workload).

{
  \bibliographystyle{plain}
  \bibliography{perf}
}

\newpage
\section*{Appendix}
\label{sec:app}

\subsection{Proof for Lemma ~\ref{lem:DelayNotDrop}}
\label{sec:applp}

\begin{IEEEproof}
Let us assume an optimal solution $\mathcal{A}$ in which there exists some control window $t_1$ (where $1 \le t_1 < T$) which violates the condition of our lemma. Therefore, of the demand $(r_{t,t+1} = p_t-a_{t,t}-d_{t,t})$ postponed (i.e., unmet) during $t_1$, there exists some portion $\nu$ (where $0 < \nu \le a_{t,t}$), that is dropped during the control windows $[t_1+1, t_1+\tau]$. Let us focus on a portion $\nu_{t_2}>0$ of $\nu$ that is dropped during the control window $t_2$ (where $t_1 < t_2 \le t_1+\tau$). Delaying $\nu_{t_2}$ over the period $(t_2-t_1)$ and dropping it during $t_2$ contributes the following to the objective: $l_{drop}(\nu_{t_2}) + l_{delay}(\nu_{t_2},t_2-t_1)$. Let $C(\mathcal{A})$ denote the objective/cost offered by $\mathcal{A}$.

Let us now compare $C(\mathcal{A})$ with the cost offered by an alternate solution $\mathcal{A'}$ which drops $\nu_{t_2}$ during $t_1$ instead of delaying it. The two algorithms' treatment of all other power demands (i.e., except for that for $\nu_{t_2}$) is exactly identical. This leads us to the following comparison of the different components of $C(\mathcal{A})$ and $C(\mathcal{A'})$:
\begin{itemize}
\item {\em Energy costs}: Since both $\mathcal{A}$ and $\mathcal{A'}$ admit the same overall energy, they have identical energy costs.
\item {\em Peak power cost}: $\mathcal{A'}$ drops $\nu_{t_2}$ {\em before} $\mathcal{A}$ does.  Consequently, the peak power consumption of $\mathcal{A'}$ cannot be worse (i.e., greater) than that of $\mathcal{A}$.
\item {\em Loss due to delaying or dropping demand}: Finally, whereas $\mathcal{A}$ incurs a cost of $l_{delay}(\nu_{t_2}, t_2-t_1) + l_{drop}(\nu_{t_2})$ for its treatment of $\nu_{t_2}$, $\mathcal{A'}$ incurs a smaller cost of $l_{drop}\nu_{t_2}$.
\end{itemize}
Combining the above, we find that $C(\mathcal{A'}) < C(\mathcal{A})$, which contradicts our assumption that $\mathcal{A}$ was optimal.
\end{IEEEproof}

\subsection{Proof for Lemma ~\ref{lem:SDPoptStructure}}
\label{sec:thresholdStructure}
\begin{IEEEproof}
If the demands $p_t$ are independent across $t$, and we define $\mu_t=\frac{a_t}{p_t},\mu_t \in [0,1]$, SDP$_{\rm Drop}$ becomes the following:
\begin{equation*}
\begin{aligned}
V_T(y_T)=&\min_{a_T,d_T} \mathbb{E} \{\alpha_T a_T + l_{\sf drop}(d_T) + \beta y_{T+1}\} \\
        =&\min_{\mu_T \in [0,1]} \mathbb{E} \{\alpha_T p_T\mu_T + l_{\sf drop}(p_T(1-\mu_T)) \\
         &+ \beta \max\{y_T,p_T\mu_T\}\}
\end{aligned}
\end{equation*}
\begin{equation*}
\begin{aligned}
V_t(y_t)=&\min_{a_t,d_t} \mathbb{E} \{\alpha_t a_t + l_{\sf drop}(d_t) + V_{t+1}(y_{t+1})\} \\
        =&\min_{\mu_t \in [0,1]} \mathbb{E} \{\alpha_t p_t\mu_t + l_{\sf drop}(p_t(1-\mu_t)) \\
         &+ V_{t+1}(\max\{y_t,p_t\mu_t\})\}, ~~t=1,...,T-1
\end{aligned}
\end{equation*}
Define $\mathcal{G}_t(\mu_t)=\mathbb{E} \{\alpha_t p_t\mu_t  +l_{\sf drop}p_t+V_{t+1}(\max\{y_t,p_t\mu_t\})\}$, and suppose $\mathcal{G}_t$ is convex, which will be proved later, and $\mathcal{G}_t$ has an unconstrained minimum with respect to $y_t$, denoted by $\phi_t$: $\phi_t=\arg \min_{\mu_t \in \mathcal{R}^+}\mathcal{G}_t(\mu_t)$. Then, in view of the constraint $0 \le \mu_t \le 1$ and the convexity of $\mathcal{G}_t$, it is easily seen that an optimal policy can determined by the sequence of scalars $\{\phi_1,\phi_2,...,\phi_T\}$ and has the form
\begin{equation*}
\begin{aligned}
\mu_t^*(y_t)=
\begin{cases}
    \phi_t, ~~&{\rm if} \phi_t \le 1\\
    1, &{\rm if} \phi_t > 1
\end{cases}
\end{aligned}
\end{equation*}
For SDP$_{\rm Drop}$, we have
\begin{equation*}
\begin{aligned}
(a_t^*,d_t^*)=
\begin{cases}
   (\phi_t p_t, p_t-\phi_t p_t), ~~&{\rm if} \phi_t \le 1\\
    (p_t, 0), &{\rm if} \phi_t > 1
\end{cases}
\end{aligned}
\end{equation*}

Now we will prove the convexity of the cost-to-go functions $V_t$ (and hence $\mathcal{G}_t$), so that the minimizing scalars $\phi_t$ exist. We use induction to prove the convexity.

For the base case, as shown above, an optimal policy at time $T$ is given by
\begin{equation*}
\begin{aligned}
\mu_T^*(y_T)=
\begin{cases}
    \phi_T, ~~&{\rm if} \phi_T \le 1\\
    1, &{\rm if} \phi_T > 1
\end{cases}
\end{aligned}
\end{equation*}
Furthermore, by plugging $\mu_T*$ back into $V_T$, we have
\begin{equation*}
\begin{aligned}
V_T(y_T)=
\begin{cases}
&\mathbb{E} \{\alpha_T p_T\phi_T + l_{\sf drop}(p_T(1-\phi_T)) \\
&~~~~~~~+ \beta \max\{y_T,p_T\phi_T\}\}, {\rm if} \phi_T \le 1\\
&\mathbb{E} \{\alpha_T p_T + \beta \max\{y_T,p_T\}\}, {\rm if} \phi_T > 1
\end{cases}
\end{aligned}
\end{equation*}
which is a convex function since $\max\{y_T,.\}$ is convex function of $y_T$. This argument can be repeated to show that for all $t=T-1,...,1$, if $V_{t+1}$ is convex, then we have
\begin{equation*}
\begin{aligned}
V_t(y_t)=
\begin{cases}
&\mathbb{E} \{\alpha_t p_t\phi_t + l_{\sf drop}(p_t(1-\phi_t)) \\
&~~~~~~~~~~+ V_{t+1}(\max\{y_t,p_t\phi_t\})\}, {\rm if} \phi_t \le 1\\
&\mathbb{E} \{\alpha_T p_T + V_{t+1}(\max\{y_t,p_t\phi_t\})\}, {\rm if} \phi_t > 1
\end{cases}
\end{aligned}
\end{equation*}
is also convex function of $y_t$. By induction, $V_t$ is convex for all $t=1,...,T$. Thus, the optimality of the above policy is guaranteed.
\end{IEEEproof}

\subsection{Proof for Lemma ~\ref{lem:onlydemanddropping}}
\label{sec:lpwithonlydemanddropping}

\begin{IEEEproof}
Without demand delaying, $r_{i,t}=0, i \in h(t) \forall t$. Therefore, we have $d_t=p_t-x_t, \forall t$. The objective becomes:
\begin{equation*}
\begin{aligned}
&\min_{a_t,y_{\rm max}} \sum_{t} \{\alpha a_t + k_{\sf drop}(p_t-a_t)\} + \beta y_{\rm max}\\
=&\min_{a_t,y_{\rm max}} \sum_{t} (\alpha-k_{\sf drop}) a_t + \beta y_{\rm max}
\end{aligned}
\end{equation*}
Since $y_{\rm max} \geq a_t, \forall t$, $y_{\rm max}$ will be equal to the largest $a_t$ in the optimal solution. For $p_t \ge y_{max}$,  we have $a_t=y_{\rm max}$, whereas for $p_t< y_{\rm max}$,  we have $a_t=p_t$. We denote this as: $a_t=p_t-(p_t-y_{\rm max})I_{\{p_t \geq y_{\rm max}\}}$, where $I_{.}$ is the standard indicator function. The optimal value is:
\begin{equation*}
\begin{aligned}
&\min_{y_{\rm max}} (\alpha-k_{\sf drop})\sum_{t} \left(p_t-(p_t-y_{\rm max})I_{\{p_t \geq y_{\rm max}\}}\right) + \beta y_{\rm max}\\
=&\min_{y_{\rm max}} (k_{\sf drop}-\alpha) \sum_{t} (p_t-y_{\rm max})I_{\{p_t \geq y_{\rm max}\}} + \beta y_{\rm max}
\end{aligned}
\end{equation*}
We denote as $V(y_{\rm max})$ the expression $(k_{\sf drop}-\alpha) \sum_{t} (p_t-y_{\rm max})I_{\{p_t \geq y_{\rm max}\}} + \beta y_{\rm max}$, and prove that $V(y_{\rm max})$ is a convex piecewise-linear function of $y_{\rm max}$. We sort the array $\{p_t\}$ into $\{\hat{p}_t\}$ such that $\hat{p}_1 \geq \hat{p}_2 \geq ... \geq \hat{p}_T$. Observing that multiple time slots might have the same power demand value, we denote $p_i^{k_i}, 1 \leq i \leq T'$ ($T'$ is the number of distinct power demand values in $\{\hat{p}_t\}$) as the $i^{\rm th}$ largest value in $\{\hat{p}_t\}$ with $k_i$ time slots having the same value $\hat{p}_i$. Then the following holds:
\begin{itemize}
\item If $y_{\rm max} > p_1^{k_1}$, none of the power demand values are larger than $y_{\rm max}$. This implies that $V(y_{\rm max})=\beta y_{\rm max}$.
\item If $p_1^{k_1} \geq y_{\rm max} > p_2^{k_2}$,  $k_1$ power demand values are larger than $y_{\rm max}$. This implies that $V(y_{\rm max})=[\beta-(k_{\sf drop}-\alpha) k_1] y_{\rm max}+(k_{\sf drop}-\alpha) k_1 p_1^{k_1}$.
\item ...
\item Finally, if $y_{\rm max} \leq p_{T'}^{k_{T'}}$, all power demand values are larger than $y_{\rm max}$. This implies that $V(y_{\rm max})=(\beta-(k_{\sf drop}-\alpha) T) y_{\rm max}+(k_{\sf drop}-\alpha)\sum_{t} p_t$.
\end{itemize}
The slope of $V(y_{\rm max})$ does not increase as $y_{\rm max}$ decreases since $k_{\sf drop}-\alpha > 0$. Therefore, $V(y_{\rm max})$ (and hence the objective function of OFF with only demand dropping) is a convex piecewise-linear function of $y_{\rm max}$. Finally, for $n=\lceil \frac{\beta}{(k_{\sf drop}-\alpha)} \rceil$, the optimal demand dropping threshold $\theta$ will be $\hat{p}_n$. 
 This optimal threshold can be found in $O(T\cdot \log T)$, the time needed for sorting the array
 $\{p_t\}$. 
\end{IEEEproof}


\subsection{Proof for Theorem~\ref{thm:cr}}
\label{sec:ONdropCR}
\begin{IEEEproof}
From the algorithm description of ON$_{\rm Drop}$, we can obtain the following properties:

{\bf Property 1:} The demand dropping threshold of ON$_{\rm Drop}$ keeps non-decreasing and is guaranteed to converge to the optimal threshold after all demands in the optimization horizon are observed.

{\bf Property 2:} The demand dropping threshold of ON$_{\rm Drop}$ never exceeds the optimal threshold. Furthermore, if we denote as $a_t^d$ the admitted demand by ON$_{\rm Drop}$ at time $t$, and $a_t^m$ the admitted demand by OFF at time $t$ when only the demand values in the first $m$ time slots are observed, $1 \le m \le T$, then $a_t^d \le a_t^m$.

It is very easy to verify the above properties by the algorithm details of ON$_{\rm Drop}$ and Lemma~\ref{lem:onlydemanddropping}. See Figure~\ref{fig:CR} for an illustration of the properties and how ON$_{\rm Drop}$ works.
\begin{figure}[htbp]
  \centering
  \includegraphics[width=0.40\textwidth]{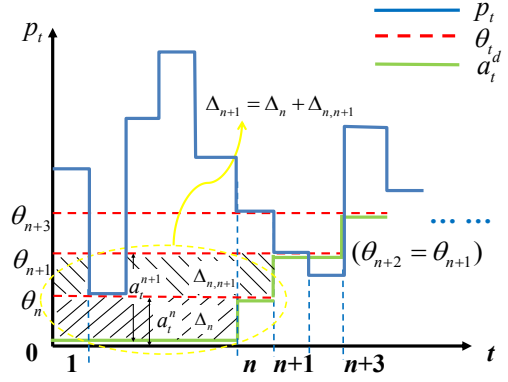}
  \caption{\small An illustration of ON$_{\rm Drop}$. $p_t$ is the original power demand; $\theta_t$ is the optimal demand dropping threshold when only the demand values in the first $t$ time slots are observed; $a_t^d$ is the admitted demand by ON$_{\rm Drop}$.}
\label{fig:CR}
\end{figure}
According to {\bf Property 1} and {\bf 2}, the total costs of ON$_{\rm Drop}$ and OFF are:
\begin{equation*}
\begin{aligned}
{\rm Cost_{OFF}}&=\beta \theta_T + \alpha \sum_{t=1}^{T} a_t^T + k_{\sf drop}\sum_{t=1}^{T} (p_t-a_t^T)\\
                &=\beta \theta_T + k_{\sf drop}\sum_{t=1}^{T} p_t - (k_{\sf drop}-\alpha) \sum_{t=1}^{T} a_t^T\\
{\rm Cost_{ON_{Drop}}}&=\beta \theta_T + \alpha \sum_{t=1}^{T} a_t^d + k_{\sf drop}\sum_{t=1}^{T} (p_t-a_t^d)\\
                      &=\beta \theta_T + k_{\sf drop}\sum_{t=1}^{T} p_t - (k_{\sf drop}-\alpha) \sum_{t=1}^{T} a_t^d
\end{aligned}
\end{equation*}

We define competitive ratio CR of ON$_{\rm Drop}$ as the upper bound of $\frac{{\rm Cost_{ON_{Drop}}}}{{\rm Cost_{OFF}}}$ under all possible workload scenarios, which is:
\begin{equation*}
{\rm CR}=\sup_{\{p_t\}_{t=1}^{T}}\frac{{\rm Cost_{ON_{Drop}}}}{{\rm Cost_{OFF}}}
\end{equation*}

Note that the only difference between the denominator and numerator is $\sum_{t=1}^{T} a_t^d$ and $\sum_{t=1}^{T} a_t^T$. From {\bf Property 2} we know that $\sum_{t=1}^{T} a_t^T \le \sum_{t=1}^{T} a_t^d$, so $\frac{{\rm Cost_{ON_{Drop}}}}{{\rm Cost_{OFF}}} \ge 1$. However, we can still find a bound $\sum_{t=1}^{T} a_t^T - \sum_{t=1}^{T} a_t^d \le (n-1)\theta_T$, $n=\lceil \frac{\beta}{(k_{\sf drop}-\alpha)} \rceil$ to make sure that CR will not go to infinity. We will prove this bound later. Now with this bound we have:
\begin{equation*}
\begin{aligned}
&\frac{{\rm Cost_{ON_{Drop}}}}{{\rm Cost_{OFF}}}\\
=&\frac{\beta \theta_T + k_{\sf drop}\sum_{t=1}^{T} p_t - (k_{\sf drop}-\alpha) \sum_{t=1}^{T} a_t^d}{\beta \theta_T + k_{\sf drop}\sum_{t=1}^{T} p_t - (k_{\sf drop}-\alpha) \sum_{t=1}^{T} a_t^T}\\
\le &\frac{\beta \theta_T + k_{\sf drop}\sum_{t=1}^{T} p_t - (k_{\sf drop}-\alpha) (\sum_{t=1}^{T} a_t^T-(n-1)\theta_T)}{\beta \theta_T + k_{\sf drop}\sum_{t=1}^{T} p_t - (k_{\sf drop}-\alpha) \sum_{t=1}^{T} a_t^T}\\
=&\frac{\beta \theta_T + k_{\sf drop}\sum_{t=1}^{T} p_t - (k_{\sf drop}-\alpha)\sum_{t=1}^{T} a_t^T + \beta\theta_T-(k_{\sf drop}-\alpha)\theta_T}{\beta \theta_T + k_{\sf drop}\sum_{t=1}^{T} p_t - (k_{\sf drop}-\alpha) \sum_{t=1}^{T} a_t^T} \\
&~~~{\rm (Since}~n=\lceil \frac{\beta}{(k_{\sf drop}-\alpha)} \rceil{\rm )}\\
=& 1 + \frac{\beta\theta_T-(k_{\sf drop}-\alpha)\theta_T}{\beta \theta_T + k_{\sf drop}\sum_{t=1}^{T} p_t - (k_{\sf drop}-\alpha) \sum_{t=1}^{T} a_t^T}\\
=& 1 + \frac{1-\frac{1}{n}}{1+\frac{k_{\sf drop}\sum_{t=1}^{T} p_t - (k_{\sf drop}-\alpha) \sum_{t=1}^{T} a_t^T}{\beta\theta_T}}\\
<& 1 + 1 -  \frac{1}{n}~~{\rm (Since}~p_t \ge a_t^T{\rm )}\\
=& 2-\frac{1}{n}
\end{aligned}
\end{equation*}

Next we prove $\sum_{t=1}^{T} a_t^T - \sum_{t=1}^{T} a_t^d \le (n-1)\theta_T \le n\theta_T$, $n=\lceil \frac{\beta}{(k_{\sf drop}-\alpha)} \rceil$ by induction. For simplicity, we define $\Delta_m=\sum_{t=1}^{m} a_t^m - \sum_{t=1}^{m} a_t^d$ and $\Delta_{m,m+1}=\Delta_{m+1}-\Delta_m$. Then our inductive hypothesis is $\Delta_m \le (n-1)\theta_m$ for $n \le m \le T$. Since the demand dropping threshold of ON$_{\rm Drop}$ is 0 for the first $n-1$ time slots, which means $a_t^d=0$ for $1 \le t \le n-1$, we choose $m=n$ as the base case.

For the base case, $m=n$, the optimal threshold is $\theta_n$ according to Lemma~\ref{lem:onlydemanddropping}, which is the $n^{\rm th}$ largest power value in the first $n$ time slots. Clearly $a_t^n=\theta_n$ for $1 \le t \le n$. For ON$_{\rm Drop}$, $a_t^d=0$ for $1 \le t \le n-1$ and $a_n^d=\theta_n$. See Figure~\ref{fig:CR} for an illustration of the base case. Therefore, $\Delta_n=\sum_{t=1}^{n} a_t^n - \sum_{t=1}^{n} a_t^d=(n-1)\theta_n$, and the inductive hypothesis holds for the base case.

Then suppose the inductive hypothesis holds for $n \le m \le i$, $i \le T-1$, which means $\Delta_m \le (n-1)\theta_m$ holds for $n \le m \le i$. By Lemma~\ref{lem:onlydemanddropping}, we know that there are at most $(n-1)$ power values that are strictly larger than $\theta_i$ (otherwise $\theta_i$ will not be the $n^{\rm th}$ largest power value in the first $i$ time slots).

When $m=i+1$, the optimal threshold becomes $\theta_{i+1}$ when only the first $i+1$ time slots are considered. If $\theta_{i+1}=\theta_i$ (only when $p_{i+1} \le \theta_i$), both the optimal solution and ON$_{\rm Drop}$ will admit $p_{i+1}$ and drop 0 demand. In that case $\Delta_{i+1}=\Delta_i \le (n-1)\theta_i = (n-1)\theta_{i+1}$, and the inductive hypothesis holds. If $\theta_{i+1}>\theta_i$ (i.e., $p_{i+1} > \theta_i$), we still have $a_{i+1}^{i+1}=a_{i+1}^d$ by Lemma~\ref{lem:onlydemanddropping}. The non-zero area between $\theta_i$ and $\theta_{i+1}$ is $\Delta_{i,i+1}$, and $\Delta_{i,i+1} \le (n-1)(\theta_{i+1}-\theta_i)$ since at most $(n-1)$ power values are strictly larger than $\theta_i$. Therefore $\Delta_{i+1}=\Delta_{i}+\Delta_{i,i+1} \le (n-1)\theta_i + (n-1)(\theta_{i+1}-\theta_i)=(n-1)\theta_{i+1}$ and the inductive hypothesis holds.

Now we can conclude that $\Delta_m \le (n-1)\theta_m$ for $n \le m \le T$, which completes our proof for the competitive ratio of ON$_{\rm Drop}$.
\end{IEEEproof}


\end{document}